\documentclass[a4paper,11pt]{article}
\pdfoutput=1

\usepackage[margin=0.8in]{geometry}
\usepackage{array}
\usepackage{bigstrut}

\usepackage{algorithm}
\usepackage{algpseudocodex}

\usepackage[english]{babel}
\usepackage{amsmath,amssymb,amsthm}
\usepackage{mathtools}

\usepackage[font=small,labelfont=bf]{caption}
\usepackage{float}
\usepackage{tikz}
\usetikzlibrary{patterns,arrows,patterns.meta}
\usepackage[font={small,sf}, labelfont=bf]{caption}

\usepackage{setspace}
\usepackage[export]{adjustbox}
\usepackage{makecell}
\usepackage[inline]{enumitem}
\usepackage{url}

\usepackage{xcolor}
\definecolor{darkblue}{RGB}{0,0,128}
\definecolor{darkgreen}{RGB}{0,150,0}

\usepackage[numbers,sort&compress]{natbib}
\usepackage[pdfusetitle]{hyperref}
\hypersetup{breaklinks, colorlinks, linkcolor=blue, citecolor=darkgreen, filecolor=red, urlcolor=darkblue}

\newtheorem{theorem}{Theorem}
\newtheorem*{theorem*}{Theorem}
\newtheorem{lemma}[theorem]{Lemma}

\newtheorem{claim}[theorem]{Claim}
\newtheorem{corollary}[theorem]{Corollary}

\newtheorem{definition}[theorem]{Definition}
\newtheorem{example}[theorem]{Example}
\usepackage{thmtools}
\usepackage{thm-restate}

\usepackage{mystyle}
\graphicspath{{../}{./}{src_tikz/}}

\usepackage{authblk}

\title{Viderman's algorithm for quantum LDPC codes}
\author{Anirudh Krishna}
\author{Inbal Livni Navon}
\author{Mary Wootters}
\affil{Stanford University, Stanford, CA, USA, 94305}
\date{}

\begin{document}

\maketitle

\begin{abstract}
    Quantum low-density parity-check (LDPC) codes, a class of quantum error correcting codes, are considered a blueprint for scalable quantum circuits.
To use these codes, one needs efficient decoding algorithms.
In the classical setting, there are multiple efficient decoding algorithms available, including \emph{Viderman's algorithm} (Viderman, TOCT 2013).  Viderman's algorithm for classical LDPC codes essentially reduces the error-correction problem to that of \emph{erasure}-correction, by identifying a small envelope $L$ that is guaranteed to contain the error set. 

Our main result is a generalization of Viderman's algorithm to {quantum} LDPC codes, namely \emph{hypergraph product codes} (Tillich, Z\'emor, IEEE T-IT, 2013).  This is the first erasure-conversion algorithm that can correct up to $\Omega(D)$ errors for constant-rate quantum LDPC codes, where $D$ is the distance of the code.
In that sense, it is also fundamentally different from existing decoding algorithms, in particular from the \texttt{small-set-flip} algorithm (Leverrier, Tillich, Z\'emor, FOCS, 2015).
Moreover, in some parameter regimes, our decoding algorithm improves on the decoding radius of existing algorithms.
We note that we do not yet have linear-time erasure-decoding algorithms for quantum LDPC codes, and thus the final running time of the whole decoding algorithm is not linear; however, we view our linear-time envelope-finding algorithm as an important first step.

\end{abstract}

\section{Introduction}

Error correcting codes play a critical role in the storage and transmission of both classical and quantum information, by protecting this information from corruption.
\emph{Low-Density Parity-Check} (LDPC) codes are a ubiquitous family of graph-based error correcting codes.
(Classical) LDPC codes were first introduced by Gallager in the 1960's~\cite{gallager1962low}, and today are used in practice for everything from satellite communication to 5G networks. 
One of the reasons for their ubiquity is that LDPC codes support extremely fast \emph{decoding algorithms}.
Such algorithms take a \emph{corrupted codeword} $\ccorruptcw \in \{0,1\}^n$ that is promised to be close in Hamming distance to a true codeword $\ccodeword \in \{0,1\}^n$, and ``correct'' $\ccorruptcw$ to return $\ccodeword$ itself.
Two of these algorithms are the $\flip$ algorithm~\cite{sipser1996expander}, and \emph{Viderman's algorithm}~\cite{viderman2013linear}.
 
Both algorithms run in linear time, but have very different flavors: $\flip$ is a greedy algorithm, which iteratively flips bits until they are all corrected; in contrast, Viderman's algorithm essentially reduces the error-correction problem to the \emph{erasure}-correction problem\footnote{In error correcting codes, an \emph{error} is when a (qu)bit has a value that is different than it is supposed to have; in particular, the decoder does not know which (qu)bits are in error. 
An \emph{erasure} is when a (qu)bit is replaced by a special symbol $\bot$. 
The erasure-correction problem is generally easier than the error-correction problem, because the decoder knows where the erasures have occured.} by identifying a small \emph{envelope} $L$ that contains all of the errors, and replacing them with $\bot$; then an erasure-correction algorithm can be run to fill in the $\bot$'s.
The two algorithms also give different guarantees: in particular, in some settings, Viderman's algorithm can correct more errors, while relaxing the constraints on the graph underlying the code construction.

Due to their importance in the classical setting, it is natural to adapt LDPC codes---and their decoding algorithms---to the quantum setting.
This was done by Tillich and Z\'emor in \cite{tillich2013quantum}, who constructed quantum LDPC Codes using \emph{hypergraph product codes}.
More recently, a series of exciting breakthroughs \cite{evra2022decodable,kaufman2021new,hastings2021fiber,breuckmann2021balanced} have led to quantum LDPC Codes ~\cite{panteleev2022asymptotically,leverrier2022quantum} with a near-optimal (up to constant factors) trade-off between the amount of redundancy required and the error correction capabilities. 

\textbf{Our question: A quantum Viderman's algorithm?}  While many constructions of quantum LDPC codes do come with linear-time decoding algorithms, all of the algorithms that we are aware of (which can correct up to $\Omega(D)$ errors) are variants of the $\flip$ algorithm mentioned above.
In particular, it has been an open question whether or not a Viderman-style algorithm---which first reduces the problem to erasure-decoding---could be used to correct a \emph{quantum} LDPC code.

\textbf{Our main result: the first Viderman-style algorithm for constant-rate Quantum LDPC Codes.} 
We discuss our results in more detail below in Section~\ref{sec:results}, but briefly, our algorithm applies to \emph{hypergraph product codes}, and identifies a small envelope $\env$ containing all of the errors in linear time.

\textbf{Why we care.} Before we discuss our results in more detail, we briefly mention a few points of motivation for a Quantum Viderman-style algorithm for hypergraph product codes.

\begin{itemize}

    \item \textbf{Converting errors to erasures.}  As discussed above, Viderman's algorithm, and our quantum analog, identifies a small \emph{envelope} $\env$, which is guaranteed to contain all of the errors.
    Then it treats the (qu)bits in $\env$ as erasures and runs an erasure-decoding algorithm to finish the job.
    In fact, coming up with such conversion algorithms for quantum codes has been a goal in other work.  For example, the \texttt{UnionFind} algorithm of Delfosse \& Nickerson gives such an algorithm for surface codes~\cite{delfosse2021union}.  Moreover, there has been some partial success for {constant-rate} quantum LDPC codes: Delfosse, Londe \& Beverland give a version of \texttt{UnionFind} that converts errors into erasures for hypergraph product codes~\cite{delfosse2022toward}.  However, the decoding radius for hypergraph product codes is $O(D^\beta)$ for some $\beta < 1$, where $D$ is the distance of the code.  In contrast, our work gives a linear-time algorithm to convert up to $\Omega(D)$ errors into erasures.
    

    \item \textbf{Improved parameters for decoding quantum LDPC codes.}
    Our quantum Viderman-style algorithm has (modestly) improved parameters over quantum $\flip$-style algorithms in certain parameter regimes (see Table~\ref{tab:comparisons} and the discussion in Section~\ref{sec:results}).
    Moreover, we are hopeful that the parameters can be further improved.
    To support our hope, we note that improvements on the original version of (classical) Viderman's algorithm have led to improved algorithms for decoding classical expander codes when the graph is a very good expander, in terms of the error radius and the requirements on the underlying graph~\cite{chen2021improved}.
    As quantum $\flip$-like algorithms have been extensively studied and optimized over the years, our algorithm provides limited improvements over the state-of-the-art $\flip$-like algorithms, but it does significantly improve the error radius over the first $\flip$-like results~\cite{leverrier2015quantum} (roughly saving a factor the underlying graph's degree, see Table~\ref{tab:comparisons}).
    Given the classical landscape,  we are hopeful that as our ideas are further developed, more significant improvements will arise.

  \item \textbf{A fundamentally new decoding algorithm.} As mentioned above, all decoding algorithms for quantum LDPC codes that we are aware of with decoding radius $\Omega(D)$ are similar to the $\flip$ algorithm of \cite{sipser1996expander}.
    Viderman's algorithm is of a fundamentally different flavor than $\flip$, and---as we discuss more below in Section~\ref{sec:techoverview}---there were several challenges to overcome to make it work in the quantum setting. 
    In addition to expanding our algorithmic toolbox for quantum LDPC codes, we hope that overcoming these challenges deepens our understanding of quantum LDPC codes, and will lead to further progress in the future.

    \item \textbf{A fundamental class of codes.}  As mentioned above, hypergraph product codes were the first quantum LDPC codes with decent distance, and remained the state-of-the-art for many years.  In recent years, a series of breakthroughs has resulted in quantum LDPC codes with much better distance.  However, we believe that hypergraph product codes are a good starting point. 

\end{itemize}

Now that we have given the high-level motivation for a quantum Viderman-style algorithm, we outline our results in more detail in Section~\ref{sec:results}. 
 After that, we will survey related work in Section~\ref{sec:related}. 

\subsection{Our Results}\label{sec:results}

Before we state our results, we introduce a bit of notation (see Section~\ref{sec:background} for full definitions).
A $\dsl N,K,D \dsr_2$ \emph{quantum error correcting code} $\qcode$ encodes a set of $K$ logical qubits in a set $\qubits$ of $N$ qubits.
It is defined by two sets $\pcs$ and $\gens$, which are sets of constraints on the qubits in $\qubits$, and which must interact with each other in a particular way.
A codeword of $\qcode$ can be thought of as an assignment of $0$ or $1$ to each qubit 
so that all of the constraints in $\pcs$ and $\gens$ are satisfied.
The \emph{distance} $D$ of $\qcode$ (formally defined in Section~\ref{sec:background}) corresponds to the number of (Pauli) errors that the code can correct.
Thus, it is desirable to have $K$ as close to $N$ as possible, while having $D$ as large as possible.

A quantum code is a \emph{quantum LDPC code} if all of the constraints in $\pcs$ and $\gens$ are \emph{sparse}, meaning that not too many qubits in $\qubits$ are involved in each one.  As with classical LDPC codes, this leads to a natural graph-theoretic connection---sparse parity-checks can be represented as sparse bipartite graphs---and the constructions of codes we consider are based on bipartite expander graphs. 

 There are two types of errors a quantum code may face, \emph{$X$-type errors} and \emph{$Z$-type errors}. 
 The constraints $\pcs$ are meant to correct $Z$-type errors, while the constraints $\gens$ are meant to correct $X$-type errors. 
 It suffices to deal with these two types individually, so for the rest of this overview we focus on $Z$-type errors. 
 We refer to $\pcs$ as ``parity-checks,'' and $\gens$ as ``generators.''  
 
Suppose that a codeword in $\qcode$ has been corrupted by a set of ($Z$-type) errors $\err \subseteq \qubits$: that is, the value of each qubit in $\err$ has been flipped.
We are guaranteed that $\err$ (after being appropriately \emph{reduced}, see Sections~\ref{sec:techoverview} and~\ref{sec:background}) is small, and we would like to correct these errors.  
The largest that $|\err|$ can be is $D/2$, half of the distance of the code, so our goal will be to allow for $|\err|$ as close to this as possible.
The number of (reduced) errors that a code can tolerate is called the \emph{decoding radius}.  

\subsubsection{Main Result: Viderman's algorithm for hypergraph product codes}

Our main result is the first quantum version of Viderman's algorithm, which we call \emph{Small-Set-Find} ($\ssfind$, which we give at a high level as Algorithm~\ref{alg:ssfind-intro} and in more detail in Algorithm~\ref{alg:ssfind}). 
$\ssfind$ applies to 
\emph{hypergraph product codes}.
Hypergraph product codes are an important class of quantum LDPC codes, and were the first quantum LDPC codes shown to achieve constant rate $K/N$ and non-trivial distance $D$~\cite{tillich2013quantum}.
We formally define them in Section~\ref{sec:HPCodes}, and for now just note that they are built out of a bipartite expander graph.

 $\ssfind$ takes as input the set $\unsat \subseteq \pcs$ of unsatisfied parity checks, and outputs an \emph{envelope} $\env$, with the guarantees that $\err \subseteq \env$ and that $\env$ is not too large.  The following is our main result.

\begin{theorem}
\label{thm:main}
Let $\mathcal{H}$ be an $\dsl N,K,D \dsr_2$ hypergraph product code (see Section~\ref{sec:HPCodes}) on qubits $\cQ$ and with $X$-parity-checks $\pcs$, constructed from an expander graph $G = (V,E)$ with left-degree $\Delta_V$ and right-degree $\Delta_C$, and expansion parameter $\epsilon < 1/10$.

Let $\err$ be a (reduced) error pattern with weight at most 
\[ |\err| < \gamma \cdot D - \Delta_V, \]
where 
$ \gamma = \frac{\Delta_V}{\Delta_C} \cdot \frac{1 - 10 \epsilon}{4}$, and let $\unsat \subseteq \pcs$ be the set of unsatisfied parity checks arising from $\err$.
Then $\ssfind$ (Algorithm~\ref{alg:ssfind}), given $\unsat$, runs in time $O_{\Delta_V, \Delta_C}(|\err|)$, and returns an envelope $\env \subseteq \mathcal{Q}$, so that $\err \subseteq \env$, and 
\[ |\env| \leq \frac{1}{\gamma} \cdot |\err|. \]
\end{theorem}

It is important to note that $\ssfind$ does not actually correct the errors $\err$; rather it reduces the problem to that of erasure-decoding, and we can use an erasure-decoding algorithm from there.  Notice that by combining the assumption on $|\err|$ with the conclusion about $|\env|$, we get that $|\env| < D$, and in particular the code can be uniquely decoded from $|\env|$ erasures.
In the classical setting, there is a relatively simple linear-time erasure-decoding algorithm that will do this~\cite{gallager1962low,shokrollahi1999new}. 
However, in the quantum setting, we unfortunately do not have a linear-time erasure-decoding algorithm.
Erasure-decoding amounts to solving a linear system and this can be done via Gaussian elimination.
The pivot points need only be qubits adjacent to the error and in this case, the size of this set is at most $N^{1/2}$.
This implies that we can do erasure decoding in time $O(N^{1.5})$.
However, at the moment, a linear-time erasure decoding algorithm  that works up to the radius $D$ remains elusive.
This is a major drawback of our result, as it means that a final decoding algorithm would run in time $O(N^{1.5})$; however, we are hopeful that follow-up work will provide a linear-time erasure-decoding algorithm, which will then result in a linear-time algorithm for decoding from errors.

In Table~\ref{tab:comparisons}, 
we compare the parameters of Theorem~\ref{thm:main} to the existing decoding algorithm for hypergraph product codes, which is called Small-Set-Flip ($\ssflip$).
This algorithm works like $\flip$, except in each greedy step instead of flipping a single bit, it flips a small set of bits at a time (hence the name).
$\ssflip$ was introduced in~\cite{leverrier2015quantum}, and further expanded in~\cite{fawzi2020constant,fawzi2018efficient,grospellier2019decodage}.\footnote{These works show that $\ssflip$ can handle a constant fraction of stochastic errors (as opposed to $O(\sqrt{N})$ adversarial errors); they also show that $\ssflip$ is robust to errors in the parity check bits.}

As we see in Table~\ref{tab:comparisons}, the parameters are in general incomparable.  Our algorithm $\ssfind$ yields an improvement in the decoding radius when $G$ is an excellent expander, and there is a large gap between $\Delta_V$ and $\Delta_C$.  For example, if $\Delta_C = 2 \Delta_V$ (so $r=1/2$) and $\epsilon = 1/20$ is small, then the decoding radius for $\ssfind$ is $0.062 D$, compared to less than $0.058D$ for $\ssflip$ (with the analysis in \cite{grospellier2019decodage}).

\setlength\extrarowheight{10pt}
\begin{table}[h]
    \centering
    
    \begin{tabular}{|m{4cm}|m{4cm}|m{2cm}|}
        \hline
         Algorithm & Decoding Radius & Required Expansion \\
         \hline
         \hline
           $\ssflip$ \cite{leverrier2015quantum} & $\frac{1}{3(1+\Delta_C)}D$       & $\epsilon < 1/6$\\
           \hline
           $\ssflip$ \cite{grospellier2019decodage}                        &  $\left(\frac{2r(1-8\epsilon)}{4+2r(1-8\epsilon)}\right)\left(\frac{r}{\sqrt{1+r^2}}\right)D$           & $\epsilon < 1/8$  \\ \hline
           \begin{minipage}{4cm} \begin{center} This paper: $\ssfind$ plus erasure-decoding   \vspace{.2cm} \end{center}\end{minipage}                      & $\frac{1-10\epsilon}{4} rD$       & $\epsilon < 1/10$\\
        \hline
    \end{tabular}
    \caption{Comparison of the parameters between different decoding algorithm of hypergraph product codes, with underlying graph is $G$, which is an $(\alpha_V,\epsilon_V,\alpha_C,\epsilon_C)$-bipartite expander (Def.~\ref{def:exp}), with left and right degree $\Delta_V$ and $\Delta_C$ respectively.
    Here we use $r=\frac{\Delta_V}{\Delta_C}$, $\epsilon = \max(\epsilon_V,\epsilon_C)$.
    Above, $D = \min(\alpha_V|V|,\alpha_C|C|)$ is the distance of the code.
    }
    \label{tab:comparisons}
\end{table}

\subsection{Related Work}\label{sec:related}

\paragraph{Algorithms for Classical LDPC Codes.}
While LDPC codes have been studied since the 1960's ~\cite{gallager1962low}, 
the first linear-time decoding algorithms for graph-based codes (or, any codes) to correct a constant fraction of \emph{worst-case} errors was given by Sipser and Spielman~\cite{sipser1996expander} who studied \emph{expander codes;} these are LDPC codes where the underlying constraint graph is an expander.\footnote{Equivalently, the code is defined as the kernel of a parity-check matrix $H$, which is the adjacency matrix of an unbalanced expander graph.}  Sipser and Spielman introduced the $\flip$ algorithm, a greedy algorithm which iteratively flips bits until it converges.
If the underlying expander $G = (V,C, E)$ is an $(\alpha_V, \eps_V)$-expander  (meaning that sets $S$ of size at most $\alpha_V n$ have neighborhoods of size at least $|S|\Delta_V(1-\eps_V)$, where $\Delta_V$ is the left-degree of the graph), the resulting expander code of length $n$ has minimum distance at least $\alpha_V n$.

When $\epsilon_V < 1/4$, Sipser and Spielman showed that $\flip$ can correct up to $\alpha_V n/(\Delta_V+1)$ errors.
Subsequent works gave other algorithms that improved this to $\left( \frac{ 1- 3\eps_V}{1 - 2\eps_V} \right)\alpha_V n$ errors, which is better when $\eps_V < 1/3$ ~\cite{feldman2006lp, viderman2013lp,viderman2013linear}; the first two of these references are based on linear programming, and the third is what we refer to as \emph{Viderman's algorithm}.  In addition to improving the decoding radius, Viderman's algorithm also requires less from the underlying expander graph (in that $\eps_V$ can be taken to be larger), making it easier to obtain constructions.  
As mentioned above, Viderman's algorithm works by identifying an envelope $\env$ of ``suspicious" bits, and then treating them as erasures.
More recently, \cite{chen2021improved} gave improved combinatorial bounds and algorithmic results for expander codes.
Their improved algorithms include variants and combinations of both $\flip$ and Viderman's algorithm; in particular, for small $\eps_V < 1/8$, they present a variant of Viderman's algorithm that decodes up to a significantly larger radius than previous works: $\frac{\sqrt{2}-1}{2\eps_V}\alpha_V n$ for very small $\eps_V$, and $\frac{1-2\eps_V}{4\eps_V} \alpha_V n$ for slightly larger $\eps_V$ that is still smaller than $1/8$. 

We note that there are several constructions of graph-based codes other than expander codes, most notably \emph{Tanner constructions}, where the code is constructed both from an expander graph and from an appropriate \emph{inner code} (for example, \cite{zemor2001expander}).
Decoding algorithms for these codes typically leverage decoding algorithms for the inner code, and thus are a bit different from the focus of our paper.

\paragraph{Algorithms for Quantum LDPC Codes.}

Until recently, it was unclear whether asymptotically good quantum LDPC codes (that is, with $K, D = \Omega(N)$) even exist.
Until 2020, hypergraph product codes were the best candidates; Tillich \& Z\'emor showed that they can achieve $K=\Theta(N)$ and $D=\Theta(\sqrt{N})$~\cite{tillich2013quantum}. 
Shortly after, Leverrier, Tillich \& Z\'emor \cite{leverrier2015quantum} proposed a quantum version of $\flip$ called Small-Set-Flip ($\ssflip$) for (expander-based) hypergraph product codes. 
The algorithm $\ssflip$ is also a linear-time algorithm, and in the adversarial setting can correct up to $\Theta(\sqrt{N})$ errors, within a constant fraction of the optimal decoding radius.

Following a series of breakthroughs crossing the $\sqrt{N}$ barrier~\cite{evra2022decodable,kaufman2021new,hastings2021fiber,breuckmann2021balanced}, asymptotically good quantum LDPC codes were finally recently attained, first by Panteleev \& Kalachev \cite{panteleev2022asymptotically}, and later simplified by Leverrier \& Z\'emor \cite{leverrier2022quantum}.
These constructions can be equipped with decoding algorithms~\cite{evra2022decodable,leverrier2023decoding,leverrier2023efficient,gu2022efficient}.
There are also variants on the code construction that have efficient decoding algorithms:
Lin \& Hsieh~\cite{lin2022good} proposed a construction based on an as-yet unresolved conjecture;
Dinur \emph{et al}.\ \cite{dinur2022good} also proposed a variant of the Tanner construction.
The decoding algorithms for all constant-rate quantum LDPC codes with distance exceeding the $\sqrt{N}$ barrier are based on generalizations of $\ssflip$.

There has been previous work on quantum error correcting codes that focuses on converting errors into erasures.
The $\mathtt{Union Find}$ algorithm by Delfosse \& Nickerson was the first erasure-conversion algorithm that runs in almost-linear-time~\cite{delfosse2021almost}.
It builds on the linear-time maximum-likelihood erasure decoding algorithm of Delfosse \& Z\'emor \cite{delfosse2020linear} for surface codes.
While $\mathtt{Union Find}$ was first proposed for the surface code, it has since been generalized to a broader class of codes \cite{delfosse2021union,delfosse2022toward}.
However, the decoding radius of $\mathtt{Union Find}$ for LDPC codes can be suboptimal.
As it applies to hypergraph product codes, the decoding radius of this algorithm was only guaranteed to scale as $\Theta(D^{\beta})$ for some constant $1 > \beta > 0$.
To the best of our knowledge, our algorithm is the first erasure-conversion algorithm for constant-rate quantum LDPC codes that achieves a decoding radius of $\Theta(D)$.

We provide a high-level contrast of $\mathtt{UnionFind}$ and Viderman's algorithm at the end of Section~\ref{sec:techoverview}.

Given an algorithm that reduces the problem to erasure-decoding, our next question is about efficient erasure-decoding algorithm.  As in the classical (linear) case, decoding a quantum error correcting code boils down to solving a linear system, and can be done straightforwardly in time $O(N^3)$.  As mentioned above, in the case of hypergraph product codes, in fact this linear system can be solved in time $O(N^{1.5})$.
In the case of stochastic erasures (where a hypergraph product code can recover from $\Omega(N)$ erasures, rather than $\Omega(\sqrt{N})$), Connolly \emph{et al}.\ ~\cite{connolly2022fast} give an improved $O(N^2)$-time erasure-decoding algorithm for hypergraph product codes, which generalizes belief propagation.

\subsection{Open Problems and Future Directions}
We view our work as the first step towards improved decoding algorithms for quantum LDPC codes.  While we are able to obtain a slight improvement in some parameter regimes for hypergraph product codes, there is still much to do:
\begin{itemize}
    \item Given that ours is only the first version of a quantum Viderman's algorithm, we hope that future work will build on our ideas to improve the parameters, obtaining bigger improvements over $\ssflip$ (and in more parameter regimes) than those reported in Table~\ref{tab:comparisons}.  In particular, we are hopeful that the ideas of \cite{chen2021improved} (which improve Viderman's algorithm in the classical case) can be applied on top of our work.
    \item We believe that our ideas can be extended to some of the more recent constructions of asymptotically good LDPC codes. 
    In particular, the construction of Lin and Hsieh~\cite{lin2022good} begins with a hypergraph product code, and quotients it out by an appropriate group action to obtain a code with better distance.
    As their construction is built directly on hypergraph product codes, this is the next natural target for our techniques\footnote{To the best of our knowledge, the constructions of \cite{panteleev2022asymptotically, leverrier2022quantum, dinur2022good} are generalizations of Tanner codes.
    For this reason, it is unclear if there is a straightforward generalization of Viderman's algorithm to these settings.}.
\end{itemize}
\subsection{Organization}
In Section~\ref{sec:techoverview}, we give a high-level technical overview of our approach, and introduce our algorithm $\ssfind$.
In Section~\ref{sec:background}, we formally give the necessary background and definitions for expander graphs; classical expander codes and Viderman's algorithm; and quantum codes and hypergraph product codes.
In Section~\ref{sec:HPCssflip}, we analyze $\ssfind$ on hypergraph product codes and prove our main theorem.

\subsection{Acknowledgements}

AK is supported by the Bloch Postdoctoral Fellowship and NSF grant CCF-1844628.
AK would also like to thank Shashwat Silas for introducing him to Viderman's algorithm and related discussions.
ILN is supported by the Simons Foundation Collaboration on the Theory of Algorithmic Fairness, the Sloan Foundation Grant 2020-13941, and the Zuckerman STEM Leadership Program.
MW was partially supported by NSF grants CCF-1844628, CCF-2231157, and CCF-2133154

\section{Technical Overview}\label{sec:techoverview}
We begin with an exposition of the classical version of Viderman's algorithm~\cite{viderman2013linear}.
Before we begin, we set up a bit more notation. 
A (classical, binary, linear) error correcting code $\ccode$ with block length $n$ is just a linear subspace of $\mathbb{F}_2^n$. 
The \emph{distance} $d$ of $\ccode$ is the minimum Hamming distance between any two distinct elements (called \emph{codewords}) of $\ccode$.
In particular, given a \emph{corrupted codeword} $\ccorruptcw \in \mathbb{F}_2^n$ with Hamming distance less than $d/2$ from some $\ccodeword \in \ccode$, the triangle inequality implies that it is in theory possible to recover $\ccodeword$.
The goal of a \emph{decoding algorithm} is to do so efficiently.

A classical \emph{LDPC code} is a code that can be defined by a sparse bipartite graph $G = (V,C, E)$ in the following way.
We associate the left-hand vertices $V$ (with $|V| = n$) with the symbols of a codeword, and the right-hand vertices $C$ with \emph{parity-checks}. 
We say that a string $\ccodeword \in \mathbb{F}_2^n$ is in $\ccode$ if, for each parity-check $c \in C$, $\sum_{v \in \Gamma(c)} \ccodeword_v = 0$, where $\nbhd(c)$ denotes the neighbors of $c$ in $G$.
We focus on the case where the underlying graph $G$ is a bipartite expander (see Definition~\ref{def:exp}); such a code is sometimes called an \emph{expander code}~\cite{sipser1996expander}.

\subsection{Classical Viderman's Algorithm}
Intuitively, Viderman's algorithm works by iteratively identifying ``suspicious'' bits and parity-checks; the suspicious bits are added to a set $L$ (called the \emph{envelope}), and the suspicious parity checks are added to a set $R$.
At the beginning, $R$ is the set $\unsat$ of unsatisfied parity checks.
From there, the rule is simple: if a vertex $v \in V$ is connected to too many suspicious checks, it, and all the checks it touches, are labeled suspicious.
This process repeats until it stabilizes (see Algorithm~\ref{alg:find-gentle}).

\begin{algorithm}[h]
    \begin{algorithmic}[0]
        \State \textbf{Input:} $\unsat \subseteq C$
        \State \textbf{Output:} $L \subseteq V$
    \end{algorithmic}
    \begin{algorithmic}[1]
        \State $h \gets (1 - 2\eps_V)\Delta_V$, \Comment{$\eps_V$, $\Delta_V$ are parameters of the underlying expander graph (see Section~\ref{sec:background}).}
        \State $L \leftarrow \emptyset$, $R \leftarrow \unsat$
        \While{$\exists v \in V$ such that $|\nbhd(v) \cap R| \geq h$}
            \State $L \leftarrow L \union \{v\}$
            \State $R \leftarrow R \union \nbhd(v)$.
        \EndWhile
        \State \Return $L$
    \end{algorithmic}
    \caption{$\find$:  Viderman's decoding algorithm.}
    \label{alg:find-gentle}
\end{algorithm}

At the end of the day, we hope that (a) the envelope $L$ contains all of the errors; and (b) $L$ is not too large.  
Viderman showed that this is indeed the case, provided that the initial number of errors in $\err$ is small enough.

\subsection{Quantum Viderman's Algorithm for Hypergraph Product Codes}
We recall the notation for a quantum error correcting code $\qcode$ from Section~\ref{sec:results}: $\qcode$ is defined by a set of parity checks $\pcs$ and generators $\gens$, and encodes $K$ logical qubits into a set $\qubits$ of $N$ qubits.
The parity checks $\pcs$ and generators $\gens$ serve to define two vector spaces (aka, classical linear codes) $\ccode_X$ and $\ccode_Z$ respectively: $\pcs$ gives the parity-checks for $\ccode_X$ and $\gens$ gives the parity-checks for $\ccode_Z$.\footnote{That is, $\ccode_X$ is the kernel of the adjacency matrix defined by $\pcs$, and similarly for $\ccode_Z$.}
For $\qcode$ to be a valid quantum code, we require $\ccode_X^{\perp} \subseteq \ccode_Z$.

\paragraph{Hypergraph Product Codes.}
We define hypergraph product codes formally in Section~\ref{sec:HPCodes}, but for now we give an informal definition and a picture.

Let $G=(V\cup C,E)$ be a biregular bipartite graph.
Our hypergraph product code will be built out of two copies of $G$.  
Qubits are associated with elements of $V \times V \sqcup C \times C$; we write $\qubits = \qubits_V \sqcup \qubits_C$ for the two parts, where $\sqcup$ denotes disjoint union.
The $X$-type parity checks $\pcs$ and $Z$-type generators $\gens$ are associated with $V\times C$ and $C \times V$ respectively.

Now, we form a graph $\cG$ on the vertices $\qubits \sqcup \pcs \sqcup \gens$, which is given by the \emph{graph product} of two copies of $G$.
The way this graph product works is formally described in Section~\ref{sec:background}, and informally illustrated in Figure~\ref{fig:hgpcode}.
For example, the qubit $(\nu,v)\in V\times V =: \qubits_V$ is connected to 
every element in 
$\set{\nu}\times\nbhd(v)\subset V\times C = \pcs$, where $\nbhd(v)$ denotes the neighborhood of $v$ in $G$. 
 In particular, there is a copy of $G$ in the row indexed by $\nu \in V$ (as well as by every other row and every column).

This graph product defines the sets $\pcs$ and $\gens$, and hence the quantum hypergraph product code $\qcode$.

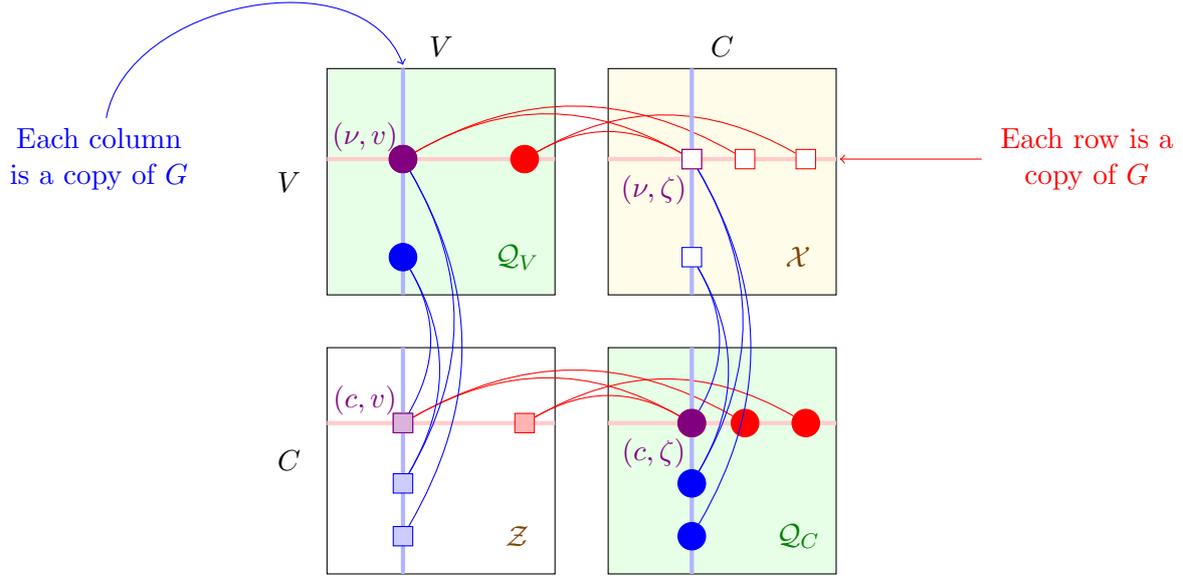
\begin{figure}[h]
    \centering
    \begin{tikzpicture}
        \draw[fill=white] (0,0) rectangle (3,3);
	\node[orange!50!black] at (2.5,.5) {$\gens$};
        \node at (-0.5,1.5) {$C$};
        \draw[fill=green!10] (0,3.7) rectangle (3,6.7);
        \node at (-0.5,5.2) {$V$};
        \draw[fill=green!10] (3.7,0) rectangle (6.7,3);
	\node[green!50!black] at (6.2,.5) {$\qubits_C$};
	\node[green!50!black] at (2.5,4.2) {$\qubits_V$};
        \node at (1.5,7) {$V$};
        \draw[fill=yellow!10] (3.7,3.7) rectangle (6.7,6.7);
	\node[orange!50!black] at (6.2,4.2) {$\pcs$};
        \node at (5.2,7) {$C$};
\begin{scope}
        \draw[ultra thick,red!20] (0,5.5) to (3,5.5);
	\draw[ultra thick,red!20] (3.7,5.5) to (6.7, 5.5);
        \node[draw,red, fill=red, circle](vv) at (1,5.5) {};
        \node[draw,red, fill=red, circle](vv2) at (2.6,5.5) {};
        \node[violet] at (.5,5.8) {$(\nu, v)$};
	\node[draw,red,fill=white,rectangle](vc) at (4.8, 5.5) {};
        \node[violet] at (4.3,5.1) {$(\nu, \zeta)$};
	\node[draw,red,fill=white,rectangle](vc2) at (5.5, 5.5) {};
	\node[draw,red,fill=white,rectangle](vc3) at (6.3, 5.5) {};

	\draw[red] (vv) to[out=30,in=150] (vc);
	\draw[red] (vv) to[out=30,in=150] (vc2);
	\draw[red] (vv2) to[out=30,in=150] (vc3);
	\draw[red] (vv2) to[out=30,in=150] (vc);

	\node[red](lblrow) at (10, 5.5) {\begin{minipage}{2.5cm}\begin{center}Each row is a copy of $G$\end{center} \end{minipage}};
	\draw[red,->](lblrow) to (6.75, 5.5);
	
	\node[blue](lblcol) at (-3, 5.5) {\begin{minipage}{2.5cm}\begin{center}Each column is a copy of $G$\end{center} \end{minipage}};
	\draw[blue,->](lblcol) to[out=80,in=110] (1, 6.75);
\end{scope}

\begin{scope}[yshift=-3.5cm]
        \draw[ultra thick,red!20] (0,5.5) to (3,5.5);
	\draw[ultra thick,red!20] (3.7,5.5) to (6.7, 5.5);
       \node[draw,violet, fill=violet!30, rectangle](vv) at (1,5.5) {};
        \node[draw,red, fill=red!30, rectangle](vv2) at (2.6,5.5) {};
	\node[draw,red,fill=red,circle](vc) at (4.8, 5.5) {};
	\node[draw,red,fill=red,circle](vc2) at (5.5, 5.5) {};
	\node[draw,red,fill=red,circle](vc3) at (6.3, 5.5) {};

	\draw[red] (vv) to[out=30,in=150] (vc);
	\draw[red] (vv) to[out=30,in=150] (vc2);
	\draw[red] (vv2) to[out=30,in=150] (vc3);
	\draw[red] (vv2) to[out=30,in=150] (vc);

\end{scope}

\begin{scope}
	\draw[ultra thick, blue!30] (1,6.7) -- (1, 3.7);
	\draw[ultra thick, blue!30] (1,0) -- (1,3);
        \node[draw,violet, fill=violet, circle](vv) at (1,5.5) {};
	\node[draw,blue,fill=blue!20,rectangle](cv) at (1, .5) {};
	\node[draw,blue,fill=blue!20,rectangle](cv2) at (1, 1.2) {};
	\node[draw,violet,fill=violet!30,rectangle](cv3) at (1, 2) {};
	\node[violet] at (.5,2.3) {$(c,v)$};
	\node[draw,blue,fill=blue,circle](vv3) at (1,4.2) {};
\draw[blue] (vv) to [out=-60,in=60]  (cv);
\draw[blue] (vv) to [out=-60,in=60]  (cv2);
\draw[blue] (vv3) to [out=-60,in=60]  (cv2);
\draw[blue] (vv3) to [out=-60,in=60]  (cv3);

\end{scope}

\begin{scope}[xshift=3.8cm]
	\draw[ultra thick, blue!30] (1,6.7) -- (1, 3.7);
	\draw[ultra thick, blue!30] (1,0) -- (1,3);
        \node[draw,violet, fill=white, rectangle](vv) at (1,5.5) {};
	\node[draw,blue,fill=blue,circle](cv) at (1, .5) {};
	\node[draw,blue,fill=blue,circle](cv2) at (1, 1.2) {};
	\node[draw,violet,fill=violet,circle](cv3) at (1, 2) {};
	\node[violet] at (.5,1.6) {$(c,\zeta)$};
	\node[draw,blue,fill=white,rectangle](vv3) at (1,4.2) {};
\draw[blue] (vv) to [out=-60,in=60]  (cv);
\draw[blue] (vv) to [out=-60,in=60]  (cv2);
\draw[blue] (vv3) to [out=-60,in=60]  (cv2);
\draw[blue] (vv3) to [out=-60,in=60]  (cv3);

\end{scope}

    \end{tikzpicture}
    \caption{Definition-by-picture of hypergraph product codes~\cite{tillich2013quantum}, formed by taking the graph product of two copies of $G$.}
    \label{fig:hgpcode}
\end{figure}

\paragraph{Reduced Errors.}
In order to explain what we mean by a ``reduced'' error in Theorem~\ref{thm:main}, we begin by pointing out what seems to be (but is not actually) a major hurdle in designing quantum LDPC codes in general, namely that the underlying graph $\cG$ is no longer a good expander.  (We note that this notion of reduced errors is not specific to our work; we include it for background and motivation).

In more detail, in the classical setting, expansion of the underlying graph $G$ guarantees that any small enough set of errors $\err$ has many \emph{unique neighbors},\footnote{A \emph{unique neighbor} of a set $\err$ is a vertex that has exactly one neighbor in $\err$.} which guarantees the existence of many unsatisfied parity checks that whenever there are not too many errors.
This guarantees that small sets of errors can be at least detected, and ideally efficiently corrected.

However, the quantum code defined by the graph product $\cG$ is \emph{not} a very good expander, and in particular there are very small sets with no unique neighbors.
To see such an example, recall that the set of generators $\gens$ is identified with $C \times V$.  Consider a particular generator $(c,v) \in \gens$.
Now consider the set 
\[\supp(c,v) = \Gamma(c) \times \{v\} \sqcup \{c\} \times \Gamma(v), \] which is the set of qubits that $(c,v)$ is adjacent to in the graph depicted in Figure~\ref{fig:hgpcode}.   This is a small set---indeed, it has constant size if the degree of $G$ is constant---but it actually has \emph{no} unique neighbors.  
To see why, we refer the reader to Figure~\ref{fig:viderman-fail-intro}(a), where we have zoomed in on $(c,v)$ and $\supp(c,v)$.  Now consider the parity checks contained in $\Gamma(c) \times \Gamma(v) \subseteq V \times C = \pcs$.  As can be seen in Figure~\ref{fig:viderman-fail-intro}(a), each such parity-check is connected to \emph{two} qubits in $\supp(c,v)$.  In particular, if $\err = \supp(c,v)$, all of the parity-checks in $\pcs$ would be satisfied.
Moreover, this is unavoidable: the requirement that $\ccode_X^\perp \subseteq \ccode_Z$ in fact necessitates this phenomenon.

At first glance then, the decoding task seems impossible, as there exist small sets of errors that cannot be detected.
However, the quantum decoding task is not identical to the classical decoding task.
In more detail, to correct a quantum error correcting code, it turns out that we do not need to be able to correct \emph{all} small sets of errors, we only need to correct errors up to ``toggling'' sets of the form $\supp(c,v)$.  (Formally, the codewords are actually \emph{cosets} modulo $\ccode_Z^\perp$, and we define the weight of a coset to be the smallest weight of any coset representative; see Section~\ref{sec:background}).  

Thus, the bad example in Figure~\ref{fig:viderman-fail-intro}(a) is not actually a bad example after all, because if we ``toggle'' the set $\supp(c,v)$, there are no errors at all.  This motivates the definition of a \emph{reduced} error, which is an error that has as small weight as possible, modulo ``toggling'' sets like $\supp(c,v)$. (That is, a reduced form of an error set $\err$ is a least weight representative in its coset modulo $\ccode_Z^\perp$).

While this turns out to not be a problem for hypergraph product codes in general, this discussion does highlight several challenges with adapting $\find$ to the quantum setting.  Below, we discuss these, and our solutions to them, in more detail.

\paragraph{First challenge: treating a single qubit as ``suspicious''.}
Given the discussion above, we now see our first  challenge in generalizing Viderman's algorithm to the quantum setting.  Viderman's algorithm ($\find$ in Algorithm~\ref{alg:find-gentle}) works by investigating each bit separately, and deciding whether it is ``suspicious'' enough---that is, has at least $h$ suspicious neighbors---to add it to the envelope $L$.  But now consider the example in Figure~\ref{fig:viderman-fail-intro}(b), where there error set $\err$ consists of \emph{half} of the set $\supp(c,v)$ for some generator $(c,v)$.  For each qubit in $\supp(c,v)$, half of the parity-checks it is connected to are satisfied, and half are unsatisfied, regardless of whether that qubit is in error or not. 
  Further, if we ``toggle'' the set $\supp(c,v)$, we get a completely different set of bad qubits that lead to the same set $\unsat$.  
  How can we identify whether a \emph{single} qubit is suspicious in this setting?  

The way we deal with this is, instead of checking for suspicious \emph{single} qubits, we check for suspicious \emph{small sets} of qubits, namely every reduced subset $\add$ in each local view $\nbhd(c) \times \nbhd(v)$. We note that this is a similar solution to how $\ssflip$ deals with the same issue.  However, in the case of Viderman's algorithm, there are a few additional challenges that we must consider.

\begin{figure}[h]
    \centering
    \begin{tikzpicture}
        \begin{scope}
        \node at (-1.4,3) {(a)};
        \draw[thick] (0,0) rectangle (3,3);

        \draw[thick] (0,-0.5)--(3,-0.5);
        \node at (1.5,-1) {$\{c\} \times \nbhd(v)$};
        \draw[thick] (-0.5,0)--(-0.5,3);
        \node[rotate=90] at (-1,1.5) {$\nbhd(c) \times \{v\}$};
        \draw (-0.3,-0.3) rectangle (-0.6,-0.6) node[below,rotate=-45] {$(c,v)$};

        \draw[thick] (-0.5,0.7) edge[-,out=30,in=150] (0.7,0.7);
        \draw[thick] (0.7,-0.5) edge[-,out=60,in=-60] (0.7,0.7);
        
        \draw[fill=white] (-0.5,0.7) circle (0.2);
        \draw[fill=white] (0.7,-0.5) circle (0.2);
        \draw[fill=white] (0.5,0.5) rectangle (0.9,0.9);
        \end{scope}
        \begin{scope}[xshift=200]
        \node at (-1.4,3) {(b)};
        \fill[gray!50] (0,0) rectangle (1.5,1.6);
        \fill[gray!50] (1.5,1.6) rectangle (3,3);
        \draw[thick] (0,0) rectangle (3,3);

        \draw[thick] (0,-0.5)--(3,-0.5);
        \node at (1.5,-1) {$\{c\} \times \nbhd(v)$};
        \draw[thick] (-0.5,0)--(-0.5,3);
        \node[rotate=90] at (-1,1.5) {$\nbhd(c) \times \{v\}$};
        \draw (-0.3,-0.3) rectangle (-0.6,-0.6) node[below,rotate=-45] {$(c,v)$};

        \draw[thick] (-0.5,0.7) edge[-,out=30,in=150] (0.7,0.7);
        \draw[thick] (0.7,-0.5) edge[-,out=120,in=240] (0.7,0.7);
        \draw[thick] (-0.5,2.25) edge[-,out=30,in=150] (1.2,2.25);
        \draw[thick] (1.2,-0.5) edge[-,out=30,in=300] (1.2,2.25);
        
        \draw[fill=white] (-0.5,0.7) circle (0.2);
        \draw[fill=red] (-0.5, 2.8) circle (0.2);
        \draw[fill=red] (-0.5, 2.3) circle (0.2);
        \draw[fill=red] (-0.5, 1.8) circle (0.2);
        \draw[fill=red] (0.2,-0.5) circle (0.2);
        \draw[fill=red] (0.7,-0.5) circle (0.2);
        \draw[fill=red] (1.2,-0.5) circle (0.2);
        \draw[fill=gray!80] (0.5,0.5) rectangle (0.9,0.9);
        \draw[fill=white] (1,2.05) rectangle (1.4,2.45);
        \end{scope}
    \end{tikzpicture}
    \caption{A view of $\supp(c,v)$ for a generator $(c,v)$, as well as the parity-checks in $\pcs = V \times C$ that are connected to $\supp(c,v)$.  (a) An illustration of the fact that the set $\supp(c,v)$ has no unique neighbors.  (b) A bad example for naively applying Viderman's algorithm.  The error $\err$ is depicted in red, and failed parity-checks are marked in gray.
    Here, \emph{all} qubits in $\supp(c,v)$ look equally ``suspicious,'' in the sense that they are connected to roughly the same number of failed parity-checks.}
    \label{fig:viderman-fail-intro}
\end{figure}
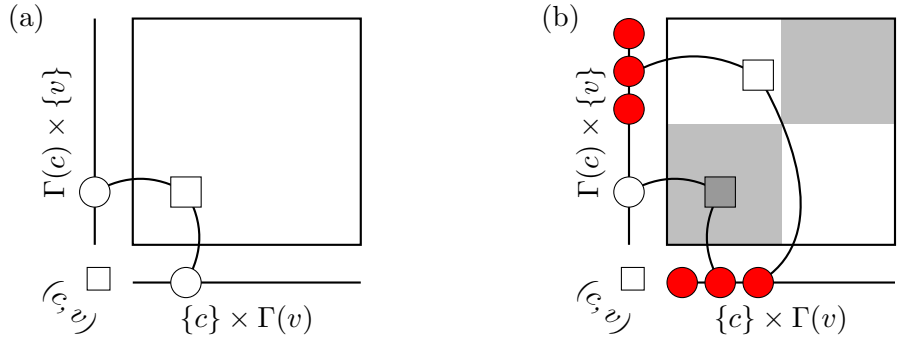

\paragraph{Second challenge: defining ``suspiciousness''.}
Once we decide to look at small sets $\add$, we still need to decide how to measure the ``suspiciousness'' of a set $\add$, to decide whether to add it to the envelope $\env$.\footnote{In the classical version of Viderman's algorithm, we use ``$L$'' and ``$R$'' as the envelope and set of suspicious checks, respectively.  To avoid confusion, we use $\env$ and $\sus$ in the quantum setting.}  A first attempt would be simply to count the number of suspicious parity checks that $\add$ is connected to (that is, the quantity $|\qnbhd(\add) \cap \sus|$, where $\qnbhd(\cdot)$ denotes the neighborhood\footnote{In the quantum setting, we use $\qnbhd$ to denote neighborhoods, and we use $\nbhd$ in the classical setting.} of $\add$, and $\sus$ denotes the set of suspicious parity checks), normalized by $|\qnbhd(\add)|$.  However, the example Figure~\ref{fig:viderman-fail-intro}(b) shows that this will probably not work.  Indeed, if $\add$ were the set of errors shown in Figure~\ref{fig:viderman-fail-intro}(b), then $|\qnbhd(\add) \cap \sus|$ is not that large, only about $2/3$ of $|\qnbhd(\add)|$.   If we set our threshold as low as $2/3$ (relative to $(1 - 2\eps_V)$ in the classical case), it seems likely that the envelope $\env$ will grow extremely large, given that the underlying graph is not a good expander. 

Instead, we define a \emph{score} function $\score(\add)$ (Definition~\ref{def:score}). 
 The main differences between the first attempt above and our score function are that we look only at \emph{unique} neighbors of $\add$ in the numerator, and we consider the intersection with the \emph{complement} $\sus^c$ instead of $\sus$.  (We also normalize things slightly differently.)  That is, instead of calling a set $\add$ suspicious when $|\qnbhd(\add) \cap \sus|$ is large (relative to $|\qnbhd(\add)|$), we instead call it suspicious if $|\uqnbhd(\add) \cap \sus^c|$ is \emph{small} (relative to some measure $\|\add\|$ of the size of $\add$).  These seem like small changes, but they have important implications.  

 In particular, sets $\add$ with lots of ``cancellations'' (that is, where $\qnbhd(\add) \setminus \uqnbhd(\add)$ is large) will in general have smaller unique neighborhoods than sets $\add$ with not as many cancellations.  Thus, sets with more cancellations will in general register as more ``suspicious'' than sets with fewer cancellations.

\begin{figure}[h]
\centering
\begin{tikzpicture}

\begin{scope}[xshift=0]
                \node at (-1.4,3) {(a)};
       \fill[orange!20] (0,0) rectangle (3,1.5);
        \draw[thick] (0,0) rectangle (3,3);
        \draw[pattern={Lines[angle=45,distance=5pt,line width=1pt]}, pattern color=blue] (1.5,0) rectangle (3,3);

        \node[blue](a) at (3.5,3) {$\sus$};
        \draw[blue,->](a) to[out=110,in=90] (2.25, 3);

        \draw[thick] (0,-0.5)--(3,-0.5);
        \node at (1.5,-1) {$\{c\} \times \nbhd(v)$};
        \draw[thick] (-0.5,0)--(-0.5,3);
        \node[rotate=90] at (-1,1.5) {$\nbhd(c) \times \{v\}$};
        \draw (-0.3,-0.3) rectangle (-0.6,-0.6) node[below,rotate=-45] {$(c,v)$};

       
         \draw[fill=orange] (-0.5, 1.3) circle (0.2);
        \draw[fill=orange] (-0.5, .8) circle (0.2);
        \draw[fill=orange] (-0.5, .3) circle (0.2);

        \end{scope}

        \begin{scope}[xshift=200]
        \node at (-1.4,3) {(b)};
        \fill[orange!30] (0,0) rectangle (1.5,1.6);
        \fill[orange!30] (1.5,1.6) rectangle (3,3);
        \draw[thick] (0,0) rectangle (3,3);

  \draw[pattern={Lines[angle=45,distance=5pt,line width=1pt]}, pattern color=blue] (1.5,0) rectangle (3,3);

        \node[blue](a) at (3.5,3) {$\sus$};
        \draw[blue,->](a) to[out=110,in=90] (2.25, 3);

        \draw[thick] (0,-0.5)--(3,-0.5);
        \node at (1.5,-1) {$\{c\} \times \nbhd(v)$};
        \draw[thick] (-0.5,0)--(-0.5,3);
        \node[rotate=90] at (-1,1.5) {$\nbhd(c) \times \{v\}$};
        \draw (-0.3,-0.3) rectangle (-0.6,-0.6) node[below,rotate=-45] {$(c,v)$};

        \draw[fill=orange] (-0.5, 2.8) circle (0.2);
        \draw[fill=orange] (-0.5, 2.3) circle (0.2);
        \draw[fill=orange] (-0.5, 1.8) circle (0.2);
        \draw[fill=orange] (0.2,-0.5) circle (0.2);
        \draw[fill=orange] (0.7,-0.5) circle (0.2);
        \draw[fill=orange] (1.2,-0.5) circle (0.2);
        
        \end{scope}
    
\end{tikzpicture}
    \caption{The two sets $\add \subset \qubits$ marked as orange circles; and their unique neighborhoods $\uqnbhd(\add) \subset \pcs$ marked as shaded orange squares.  The set $\sus$ of ``suspicious'' parity checks is drawn in blue lines.  The set $\add$ in (b) will count as much more ``suspicious'' with our $\score$ function than the one in (a), intuitively because it has more ``cancellations''.}\label{fig:overlaps}
\end{figure}
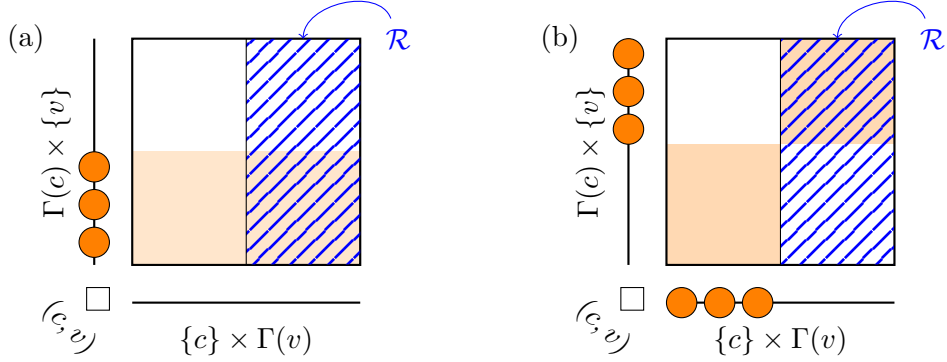

 To see an example of this, consider the two sets $\add$ shown in Figure~\ref{fig:overlaps} (a) and (b).
The set shown in (a) would count as \emph{more} suspicious than the set in (b) with either of the score functions $|\qnbhd(\add) \cap \sus|/|\qnbhd(\add)|$ or $|\uqnbhd(\add) \cap \sus|/|\qnbhd(\add)|$ (and would also be more suspicious if we normalized by $\|\add\|$, which is what we actually do in our score function).  However, if we ask that  $|\uqnbhd(\add) \cap \sus^c|/|\qnbhd(\add)|$ is small, the set shown in (a) would actually count as \emph{less} suspicous than the set in (b) (and it would be the same if we normalized by $\|\add\|$).  Thus, our score function has the desired behavior while the other two candidates do not: sets with more ``cancellations'' should count as \emph{more} suspicious.

We remark that we view coming up with the ``right'' $\score$ function to be one of the main contributions of our work.  As the discussion above shows, there are many options, and subtle differences can be important.

\paragraph{Interlude: the algorithm!} With these first two challenges and solutions described, we can now state an informal version of $\qfind$ (see Algorithm~\ref{alg:ssfind} for the formal version).
Intuitively, our algorithm works the same way as the classical Viderman's algorithm $\find$, except that (1) we consider small sets $\add$ rather than single qubits when updating $\env$; and (2) we use our score function $\score$ to decide if a set $\add$ should be added to $\env$.

\begin{algorithm}[H]
    \begin{algorithmic}[0]
        \State \textbf{Input:} $\unsat \subseteq \pcs$
        \State \textbf{Output:} $\env \subseteq \qubits$.
    \end{algorithmic}
    \begin{algorithmic}[1]
        \State $\env \leftarrow \emptyset$
        \State $\sus \leftarrow \unsat$
        \While{$\exists$ an appropriate set $\add$ such that $\score(\add) \leq h$}
            \State $\env \leftarrow \env \union \add$.
            \State $\sus \leftarrow \sus \union \qnbhd(\add)$
        \EndWhile
        \State \Return $\env$.
    \end{algorithmic}
    \caption{High-level description of $\ssfind$; see Algorithm~\ref{alg:ssfind} for the full version.}
    \label{alg:ssfind-intro}
\end{algorithm}

\paragraph{Third Challenge: the analysis.} 
Once we have our candidate algorithm, the analysis turns out to be much more delicate than the analysis of the classical Viderman's algorithm ($\find$).  We give a brief overview of the proof structure, and the challenges that arise.

First, we recap the analysis of $\find$ in the classical case.
For completeness, we present a simplified version of the classical analysis (simpler than in~\cite{viderman2013linear}, but with a worse quantitative result) in Appendix~\ref{sec:classic-code}.
We give a high-level overview here, and we suggest that the reader go through Appendix~\ref{sec:classic-code} before reading our proof of correctness for the quantum algorithm.

Intuitively, the analysis of $\find$ proceeds in two steps, \emph{coverage} and \emph{bounded growth}.
In the coverage step, we need to show that the envelope  eventually contains all of the errors $\err$.
In the bounded growth step, we need to show that the envelope does not grow too large.  
Below, we discuss both steps, and explain why they are difficult to generalize to the quantum setting.

\emph{Challenge 3a: Coverage Step.}  At a very high level, the argument for coverage in the classical case is as follows.
Suppose that the envelope $L$ does not cover all of $\err$, and let $B = \err \setminus L$ be the part that has not been covered.
By expansion of the underlying graph, the unique neighborhood $\unbhd(B)$ of $B$ is large.
By an averaging argument, there exists some $v \in B$ with many neighbors in $\unbhd(B)$.  However, we claim that  $\unbhd(B) \subseteq R$: that is, every element of $\unbhd(B)$ has already been labeled as suspicious.
Indeed, either these elements have another neighbor in $\err \setminus B$, in which case they were already labeled as suspicious since $\err \setminus B \subseteq L$ has already been labeled suspicious; or they do not, in which case they were in $\unbhd(\err)$ and hence were in $\unsat$, which was labeled suspicious at the beginning.
But then $v$ has many neighbors in $R$, meaning it should have already been added to $L$, a contradiction.

While the coverage proof in the classical case is quite simple, in the quantum case things get much more complicated.
It is still true, and not hard to see, that $\uqnbhd(\mathcal{B}) \subseteq \sus$ (where we use the caligraphic letters $\mathcal{B}$ and $\sus$ to represent the quantum analogs of $B$ and $R$).
However, we are no longer easily guaranteed the existence of a qubit $\qubit \in \mathcal{B}$ that touches many of these vertices, since our graph does not have expansion.
Instead, we employ a delicate argument that leverages expansion of the rows and columns separately.
This argument is handled in Section~\ref{subsec:coverage}.

\emph{Challenge 3b: Bounded growth step.}  Again, we recap the high-level argument in the classical case.  The basic idea is to bound $|\nbhd(L)|$, the size of the neighborhood of the envelope, in two ways.  First, by expansion, $|\nbhd(L)|$ must be large if $L$ is large.  On the other hand, consider building $L$ one vertex at a time.  Each time we add a vertex $v$ to $L$, how much can $|\nbhd(\env)|$ change?  Intuitively, this is not too much, because we only add vertices to $L$ because they have many neighbors in $R$; but a large part of $R$ made up of $\nbhd(L)$ (the rest is $\unsat$).  Thus, any vertex has many neighbors in $R$, hence many neighbors in $\nbhd(L)$, and so does not add too much to $\nbhd(L)$.  This gives an upper bound on $|\nbhd(L)|$.  If $L$ gets too large, these two bounds yield a contradiction, completing the argument.

We mirror the same basic approach---proving both an upper bound and a lower bound---in the quantum setting, but again it is now much more difficult.  For the lower bound, we no longer have good expansion, so we cannot argue immediately that $|\qnbhd(\env)|$ is large just because $\env$ is large; however, it turns out that we can do this by again leverage the fact that both the rows and columns expand.  For the upper bound, one challenge comes from our change to the score function.  As discussed above, intuitively the score function says that a set with more ``cancellations'' should be more suspicious; thus, relative to the classical argument, we may be adding ``suspicious'' sets whose neighborhoods do not overlap as much with $\sus$.  However, we can again use the expansion of the rows and columns to show that this term is manageable.

To complete the bounded growth step, which is handled in Section~\ref{subsec:bounded-L}, we put together the upper bound and the lower bound on $|\qnbhd(\env)|$, and conclude that $\env$ has not grown too large.  This then completes our proof of correctness.

\subsection{Comparison with \tt{UnionFind}} As $\mathtt{Union Find}$ is also an erasure-conversion algorithm that can be applied to hypergraph product codes~\cite{delfosse2022toward} (although with asymptotically smaller error radius, as noted above), we briefly describe it to contrast it with $\ssfind$.
$\mathtt{UnionFind}$ iteratively maintains and updates a set of disjoint clusters $\{\mathsf{Clust_i}\}$ in the subgraph of $\cG$ induced by $\qubits \sqcup \pcs$.
Each cluster $\mathsf{Clust}_i$ forms a connected component in the graph that includes both parity checks and qubits.
The clusters are initialized as the neighborhoods of $\unsat$.
In each iteration, \texttt{UnionFind} tries to find errors $\ferr_i$ in each cluster $\mathsf{Clust}_i$ such that the syndrome of $\ferr_i$ is equal to the unsatisfied parity checks in the interior of $\mathsf{Clust}_i$.
If no such error is found, it enlarges $\mathsf{Clust}_i$ by adding to it all neighbors of $\mathsf{Clust}_i$.
If two clusters $\mathsf{Clust}_i$ overlap, then they are merged.
The main difference between the $\mathtt{UnionFind}$ algorithm and a Viderman-style algorithm is that Viderman's algorithm does not add the full neighborhood of suspected parity-checks, but only those that are ``suspicious enough''.

The clusters $\mathsf{Clust}_i$ can be compared to the union of the envelope $\env$ and the suspicious parity checks $\sus$.
In contrast to $\mathtt{UnionFind}$, $\ssfind$ only adds sets of qubits $\add$ to $\env$ if they have a sufficiently low score.
Furthermore, these sets are subsets of the support $\supp(\gen)$ for generators $\gen \in \gens$.
Finally, $\ssfind$ does not verify whether the envelope contains an error in each iteration; rather, it proceeds until there are no more sets of qubits with sufficiently low score.

\section{Background and Definitions}\label{sec:background}
Before we proceed to the proof, we give the formal definitions that we will need.
\subsection{Expander graphs}
Let $V = [n]$ and $C = [m]$, and $G= (V \union C, E)$ be an undirected bipartite graph that is biregular with left node degree $\Delta_V$ and right node degree $\Delta_C$.

\begin{definition}[Graph neighborhood]
For $S_V \subseteq V$ and $S_C \subseteq C$, let $E(S_V,S_T)$ denote the set of edges between $S_V$ and $S_T$.

For $S_V \subseteq V$, we denote by $\nbhd(S_V)$ the neighborhood of $S_V$
\begin{align}
    \nbhd(S_V) = \{c \in C\;|~ |E(\{c\}, S_V)| \geq 1\}~.
\end{align}
For $S_V \subseteq V $, $\unbhd(S_V)$ denotes the unique neighborhood of $S$
\begin{align}
    \unbhd(S_V) = \{c \in C \;|~ |E(\{c\}, S_V)| = 1\}~.
\end{align}
The definition of neighborhood of $S_C\subset C$ is analogous.
\end{definition}
\begin{definition}[Vertex expander]\label{def:exp}
We say $G$ is an $(\alpha_V, \epsilon_V)$ left vertex-expander if for all $S_V\subseteq V$,
\begin{align}
    |S_V| \leq \alpha_V n \implies |\nbhd(S_V)| \geq (1-\epsilon_V)\Delta_V|S_V|~.
\end{align}
The graph $G$ is an $(\alpha_C,\epsilon_C)$ right vertex-expander if
\begin{align}
    |S_C| \leq \alpha_C m \Longrightarrow |\nbhd(S_C)| \geq (1-\epsilon_C)\Delta_C |S_C|~.
\end{align}
The graph $G$ is an $(\alpha_V, \epsilon_V,\alpha_C,\epsilon_C)$ bidirectional vertex expander if it is an $(\alpha_V, \epsilon_V,)$ left vertex-expander and $(\alpha_C,\epsilon_C)$ right vertex-expander.
\end{definition}
The parameter $\epsilon_V$ is related to the number of collisions between outgoing edges of $S_V\subset V$, i.e.\ vertices in $C$ where multiple edges from $S_V$ are incident. Specifically, by a simple averaging argument we can deduce that for $S_V\subset V$, if $|\nbhd(S_V)| \geq (1-\epsilon_V)\Delta_V|S_V|$ then also $|\unbhd(S_V)| \geq (1-2\epsilon_V)\Delta_V|S_V|$.

\subsection{Classical expander codes}
Expander codes were introduced by Sipser and Spielman \cite{sipser1996expander}.  They are built from a good vertex expander, and defined as follows. 
\begin{definition}[Expander code]\label{def:exp-code}
Let $G=(V\cup C,E)$ be a $(\Delta_V,\Delta_C)$ bi-regular bipartite graph, which is an $(\alpha_V,\epsilon_V)$ left-vertex expander. Denote $|V|=n,|C|=m$.

The expander code induced by the graph $G$ is the subspace $\ccode\subset\set{0,1}^n$ such that for every $\ccodeword\in \ccode$ and $c\in C$,
\[ \bigoplus_{v\in\nbhd(c)} w_v = 0~, \]
where $\oplus$ is the {\rm XOR} function.

\end{definition}
As discussed earlier,
Sipser and Spielman suggested a decoding algorithm $\flip$ for expander codes, based on bit flips. Viderman \cite{viderman2013linear} suggested a different decoding algorithm, $\find$, which finds a set of ``suspicious'' bits $L$ that are guaranteed to contain the error $\err$.
Viderman shows that  as long as the parameter $\eps_V \leq 1/3$ and the number of errors is at most
\[ |\err| \leq \left(\frac{1-3\eps_V}{1 - 2\eps_V}\right)n, \]
then $\find$ is correct, and returns a small envelope. 

We have already given a version of the $\find$ algorithm in Algorithm~\ref{alg:find-gentle}. 
For completeness, in Appendix~\ref{sec:classic-code} we give a slightly simpler proof than the one in \cite{viderman2013linear} for the correctness of Algorithm~\ref{alg:find-gentle}. We obtain worse guarantees, but this slightly simpler proof (which is implicit in \cite{viderman2013linear}) is the one we generalize.

\subsection{Quantum codes}\label{sec:quantum-codes}

Given a classical code $\ccode \in \bbF_2^{N}$, we let $\ccode^{\perp} = \{\bu:\; \forall \bv \in \ccode,\; \sum_i u_i v_i = 0 \pmod{2} \}$ denote its dual.
A $\dsl N, K, D\dsr$ CSS quantum code $\qcode$ on a set of $N$ qubits $\qubits$ is specified by two classical codes $\ccode_Z,\ccode_X \subseteq \bbF_2^N$ that obey the relation $\ccode_X^{\perp} \subseteq \ccode_Z$.
The codes $\ccode_X^{\perp}$ and $\ccode_Z^{\perp}$ form the parity checks for the quantum code.
We let $\pcs$ and $\gens$ be the basis for $\ccode_X^{\perp}$ and $\ccode_Z^{\perp}$ respectively.

The number of encoded qubits $K$ is defined as
\begin{align}
    K = \dim(\ccode_X/\ccode_Z^{\perp}) = \dim(\ccode_X) + \dim(\ccode_Z) - N~.
\end{align}
The distances $D_Z$ and $D_X$ are defined as
\begin{align}
    D_Z = \min\{ |\mathbf{e}| : \mathbf{e} \in \ccode_X \setminus \ccode_Z^{\perp}\}~, \quad
    D_X = \min\{ |\mathbf{e}| : \mathbf{e} \in \ccode_Z \setminus \ccode_X^{\perp}\}~.
\end{align}
The distance of the code $\qcode$ is defined as $D = \min(D_X,D_Z)$.

We identify the set of qubits $\qubits$ with the set $\{1,...,N\}$ and vectors $\mathbf{e} \in \bbF_2^N$ with their support $\err \subseteq \qubits$.
Addition of vectors mod $2$ corresponds to the symmetric difference of sets, denoted $\oplus$: For $\err_1, \err_2 \subseteq \qubits$, we have
$\err_1 \oplus \err_2 = [\err_1 \setminus \err_2] \union  [\err_2 \setminus \err_1]$.

For $\gen \in \gens$ and $\pc \in \pcs$, $\supp(\gen) \subseteq \qubits$ and $\supp(\pc) \subseteq \qubits$ denote the supports of the corresponding basis element.

It turns out that $X$ and $Z$ errors can be corrected separately.
As the two cases mirror each other, we will only consider $Z$-type errors.
The $X$-type parity checks are used to correct them and we shall refer to these as {parity checks}.
In this case, we are interested in the syndromes for the $X$-type parity checks.
For an error $\err$, we let $\unsat$ denote the $X$-type syndrome (unsatisfied $X$-type parity checks):
\begin{align}
    \unsat = \set{\pc | |\supp(\chi) \cap \err| \neq 0 \pmod{2}}~.
\end{align}
The $Z$-type parity checks are not treated as constraints in this view; rather they define $\ccode_Z^{\perp}$, the set of words equivalent to the trivial codeword.  We refer to the set $\gens$ of $Z$-type parity checks as \emph{generators}.

\subsection{Hypergraph Product Codes}
\label{subsec:background-quantum-expander}
\label{sec:HPCodes}

Let $G = (V \union C, E)$ be a $(\Delta_V, \Delta_C)$-biregular bipartite graph.
Suppose it is an $(\alpha_V,\epsilon_V,\alpha_C,\epsilon_C)$ bidirectional vertex-expander, with
$|V| = n$ and $|C| = m$. Recall that we denote by $\nbhd$ the neighborhood in the graph $G$. 
We assume that the graph is unbalanced, i.e. that $n>m$ and $\Delta_C>\Delta_V$.

The $[n,k,d]$ code $\ccode(G)$ is the code defined on Definition~\ref{def:exp-code}, where $V$ and $C$ are associated with bits and parity checks respectively.
The $[m,\tk,\td]$ code $\tC(G)$ is the code associated with the graph $G$, in which $C$ and $V$ swap roles: $V$ is associated with parity checks and $C$ with the bits.
For convenience, we say that $d = \infty$ ($\td = \infty$) if $k = 0$ ($\tk = 0$).

The hypergraph product code $\qcode$ is an $\dsl N,K,D\dsr$ code where $N = n^2 + m^2$, $K = (k)^2 + (\tk)^2$ and $D = \min(d,\td)$.\footnote{Technically the actual distance may be larger than $\min(d,\td)$, which is the \emph{design} distance, so we should have written $D \geq \min(d,\td)$. 
However, for simplicity, we will refer to the design distance as simply the ``distance,'' and denote it by $D$.}
For the proof of these facts, we point the reader to the original paper by Tillich and Z\'emor \cite{tillich2013quantum}.
We proceed to describe the code construction.

\begin{definition}
Let $G_1 = (V_1 \cup C_1, E_{1})$ and $G_2 = (V_2 \cup C_2, E_{2})$ denote two copies of $G$.
The hypergraph product code $\qcode$ is constructed from the graph product $\cG := G_1 \times G_2$, where,
\begin{enumerate}
    \item The set of qubits, denoted $\qubits$, is associated with $\qubits_V \sqcup \qubits_C$, where $\qubits_V = V_1 \times V_2 $ and $\qubits_C = C_1 \times C_2$.
    \item The set of $Z$ stabilizer generators, denoted $\gens$, is associated with $C_1 \times V_2$.
    We shall let $\gen$ denote a generator and write $\gen \sim (c,v) \in C_1 \times V_2$ to say $\gen$ is identified with the tuple $(c,v)$.
    The support $\supp(\gen)$ of the generator $\gen$ is 
    \begin{equation}
        \supp(\gen) = \{(\nu,v): \nu \in \Gamma(c)\} \union \{(c,\zeta): \zeta \in \Gamma(v)\}~.
    \end{equation}
    \item The set of $X$ stabilizer generators, denoted $\pcs$, is associated with $V_1 \times C_2$.
    For $\nu \in V_1$ and $\zeta \in C_2$, the support $\supp(\pc)$ of a parity check $\pc \sim (\nu,\zeta)$ is 
    \begin{equation}
        \supp(\pc) = \{(\nu,v) : v \in \Gamma(\zeta)\} \union \{(c,\zeta) : c \in \Gamma(\nu)\} ~.
    \end{equation}
\end{enumerate}
\end{definition}
We refer to Fig.\ \ref{fig:hgpcode} for an illustration.

Throughout, we use Roman letters for $C_1$ and $V_2$, and Greek letters for $C_2$ and $V_1$.  Thus, generators are Roman tuples like $(c,v)$, while parity checks are Greek tuples like $(\nu, \zeta)$.

We note that in this case, the codes $\ccode_Z^{\perp}$ and $\ccode_X^{\perp}$ are given by
\begin{align*}
    \ccode_Z^{\perp} = {\rm span}(\supp(c_1,v_2) : c_1 \in C_1, v_2 \in V_2)~, \quad
    \ccode_X^{\perp} = {\rm span}(\supp(\nu_1,\zeta_2) : \nu_1 \in V_1, \zeta_2 \in C_2)~.
\end{align*}

If we ignore the generators $\gens$, we get a bipartite graph between the qubits $\qubits$ and then $X$-type parity checks $\pcs$ (that is, the induced graph of $\cG$ on the vertices $\qubits \sqcup \pcs$).  We denote this graph by $\cG_X$.

To make the notations clearer, we use different symbols the neighborhood in the classical graph $G$ and in the quantum graph $\cG$.  In the classical setting, we use $\nbhd$, while in the quantum setting, we use $\qnbhd$.  Formally, we have the following definitions.
\begin{definition}[Quantum neighborhood]\label{def:qgraph}
Let $G=(V\cup C,E)$ be a graph and $\cG_{X}$ be the graph induced by qubits $\qubits$ and parity checks $\pcs$.
For a qubit $\qubit \in \qubits$, we denote its neighborhood by $\qnbhd(\qubit)$, defined as follows:
\begin{align*}
    &\qnbhd(\set{\nu_1,v_2}) = \set{(\nu_1,\zeta_2)| \zeta_2\in\nbhd(v_2)}~, \quad 
    &\qnbhd(\set{c_1, \zeta_2}) = \set{(\nu_1,\zeta_2)| \nu_1\in\nbhd(c_1)}~,
\end{align*}
where $\nbhd$ is the neighborhood in the classic graph $G$.

This naturally extends to a set of qubits, $\qset \subseteq \qubits$, $\qnbhd(\qset) = \cup_{q\in\qset}\qnbhd(\qset)$. 

The unique neighborhood $\uqnbhd(\qset)$ and multineighborhood $\mqnbhd(\qset)$ mirror their classical definitions:
\begin{align*}
    \uqnbhd(\qset) = \set{ \pc | \text{ there is a unique }  \qubit \in \qubits \text{ such that } \pc \in \qnbhd(\qubit)}~, \quad
    \mqnbhd(\qset) = \qnbhd(\qset) \setminus \uqnbhd(\qset)~.
\end{align*}
\end{definition}

For a generator $\gen \in \gens$, we shall abuse notation to write $\qnbhd(\gen)$ to refer to $\qnbhd(\supp(\gen))$.

In addition, we will use the following notation:
\begin{itemize}
    \item For a set of qubits $\qset\in\qubits$, let $\qset_V = \qset\cap \qubits_V$ and $\qset_C = \qset\cap \qubits_C$.
\item Let $\Delta = \Delta_C\times\Delta_V$.
\item Let $\epsilon = \max(\epsilon_C,\epsilon_V)$.
\end{itemize}

Finally, we define a notion that will be useful in our analysis.
\begin{definition}[Weighted norm]\label{def:wtd_norm}
    For $\qset\subseteq \qubits$, we define $\lVert \qset \rVert = \frac{|\qset_V|}{\Delta_C} + \frac{|\qset_C|}{\Delta_V}$. 
\end{definition}

\subsection{Projections of sets}

In this section, we introduce some notation for referring to projections of sets on to the component graphs $G_1$ and $G_2$.
This notation will be used for all projections throughout this paper.

\begin{definition}[Projection]\label{def:proj}
Let $\cG = (\qubits,\pcs)$ be a graph that is a product of $G_1$ and $G_2$.
    
For all $\nu_1 \in V_1$ and all $v_2 \in V_2$, let $\qsetvtwo\subset V_1,\qsetvone\subset V_2$ be the projections on to $G_1$ and $G_2$ respectively:
\begin{align*}
    \qsetvtwo = \set{\nu_1' | (\nu_1',v_2) \in \qset_V}~, \qquad \qsetvone = \set{v_2' | (\nu_1,v_2') \in \qset_V}~.
\end{align*}
We obtain $\qsetVone\subset V_1$ and $\qsetVtwo\subset V_2$ as the union of sets
\begin{align}\label{eq:proj}
    \qsetVone = \bigcup_{v_2} \qsetvtwo~,\qquad \qsetVtwo = \bigcup_{\nu_1} \qsetvone~.
\end{align}
Similarly, for all $c_1 \in C_1$ and $\zeta_2 \in C_2$, let $\qsetcone\subset C_2$ and $\qsetctwo\subset C_1$ denote the projections on to $G_1$ and $G_2$ respectively:
\begin{align}
    \qsetcone = \set{\zeta_2' | (c_1,\zeta_2') \in \qset_C}~,\quad
    \qsetzetatwo = \set{c_1' | (c_1', \zeta_2) \in \qset_C}~.
\end{align}
We obtain $\qsetCone\subset C_1$ and $\qsetCtwo\subset C_2$ as the union of sets
\begin{align}    
    \qsetCtwo = \bigcup_{c_1 \in C_1} \qsetcone~,\quad\qsetCone = \bigcup_{\zeta_2 \in C_2} \qsetzetatwo~.
\end{align}
\end{definition}

We refer to Figure~\ref{fig:projections} for a schematic.

\begin{figure}[h]
\centering
\begin{tikzpicture}[scale=0.8]
    \begin{scope}
        \draw (-1,-3) rectangle (4,2); 
        \draw (-1,2.2) rectangle (4,7.2); 
        \draw (4.2,-3) rectangle (9.2,2); 
        \draw (4.2,2.2) rectangle (9.2,7.2); 
        
        \node at (-1.5,-0.5) {$C_1$};
        \node at (-1.5,4.7) {$V_1$};
        \node at (1.5,-3.5) {$V_2$};
        \node at (6.5,-3.5) {$C_2$};
        
        \foreach \x in {0.8} {
            \foreach \y in {4.1,4.4,4.7}
                \draw[orange,fill=orange!30] ({\x+0.15},{\y+0.15}) circle (0.145);
        }
        \foreach \x in {1.1} {
            \foreach \y in {3.8,4.1,4.4,4.7,5}
                \draw[orange,fill=orange!30] ({\x+0.15},{\y+0.15}) circle (0.145);
        }
        \foreach \x in {1.4} {
            \foreach \y in {3.8,4.1,4.7}
                \draw[orange,fill=orange!30] ({\x+0.15},{\y+0.15}) circle (0.145);
        }
        \draw[{|-|},purple] (0.8,5.5)--node[above] {$\qsetVtwo$} (1.7,5.5);
        \draw[{|-|},purple] (0.6,3.8)--node[left] {$\qsetVone$} (0.6,5.3);
        
        \foreach \x in {5.5} {
            \foreach \y in {3.8,4.1,4.4}
                \draw[orange,fill=orange!30] (\x,\y) rectangle ({\x+0.3},{\y+0.3});
        }
         \foreach \x in {5.8} {
            \foreach \y in {4.4,4.7}
                \draw[orange,fill=orange!30] (\x,\y) rectangle ({\x+0.3},{\y+0.3});
        }
        \foreach \x in {6.1} {
            \foreach \y in {3.8,4.1,4.4,4.7,5}
                \draw[orange,fill=orange!30] (\x,\y) rectangle ({\x+0.3},{\y+0.3});
        }
        \foreach \x in {6.4} {
            \foreach \y in {3.8,4.1,4.7}
                \draw[orange,fill=orange!30] (\x,\y) rectangle ({\x+0.3},{\y+0.3});
        }
        \foreach \x in {6.7} {
            \foreach \y in {4.1,4.4}
                \draw[orange,fill=orange!30] (\x,\y) rectangle ({\x+0.3},{\y+0.3});
        }
        \foreach \x in {7} {
            \foreach \y in {4.1,4.4}
                \draw[orange,fill=orange!30] (\x,\y) rectangle ({\x+0.3},{\y+0.3});
        }
        \foreach \x in {7.3} {
            \foreach \y in {4.4,4.7}
                \draw[orange,fill=orange!30] (\x,\y) rectangle ({\x+0.3},{\y+0.3});
        }
        \draw[{|-|},purple] (5.5,5.5)--node[above] {$\nbhd(\qsetVtwo)$} (7.6,5.5);
        \draw[{|-|},purple] (1.1,3.5)--node[below] {$\qsetvtwo$} (1.7,3.5);
        \draw[{|-|},purple] (5.5,3.5)--node[below] {} (5.8,3.5);
        \draw[{|-|},purple] (6.1,3.5)--node[below] {} (6.7,3.5);
        \node[purple] at (6.1,3.1) {$\nbhd(\qsetvtwo) $};
        \node [gray] at (-1.3,3.9) {$v_2$};
        \draw[gray,dashed] (-1,3.95) to (9.2,3.95);
        \foreach \y in {-0.8} {
            \foreach \x in {6.1,6.7}
                \draw[orange,fill=orange!30] ({\x+0.15},{\y+0.15}) circle (0.145);
        }
        \foreach \y in {-1.1} {
            \foreach \x in {6.4,6.7,7}
                \draw[orange,fill=orange!30] ({\x+0.15},{\y+0.15}) circle (0.145);
        }
        \foreach \y in {-1.4} {
            \foreach \x in {6.1,6.7}
                \draw[orange,fill=orange!30] ({\x+0.15},{\y+0.15}) circle (0.145);
        }
        \foreach \y in {-0.8} {
            \foreach \x in {5.5}
                \draw[orange,fill=orange!30] ({\x+0.15},{\y+0.15}) circle (0.145);
        }
        \foreach \y in {-1.1} {
            \foreach \x in {5.5}
                \draw[orange,fill=orange!30] ({\x+0.15},{\y+0.15}) circle (0.145);
        }
        \foreach \y in {-1.4} {
            \foreach \x in {5.8}
                \draw[orange,fill=orange!30] ({\x+0.15},{\y+0.15}) circle (0.145);
        }
        \draw[{|-|},purple] (5.5,-1.9)--node[below] {$\qsetCtwo$} (7.3,-1.9);
        \draw[{|-|},purple] (7.5,-1.7)--node[right] {$\qsetCone$} (7.5,-0.5);

        \draw[{|-|},purple] (7.8,3.8)--node[rotate=90,below] {$\nbhd(\qsetCone)$} (7.8,5.3);
   \end{scope}
\end{tikzpicture}
\caption{Qubits that form $\qset$ are denoted using circular nodes in $V_1 \times V_2$ and $C_1 \times C_2$.
The projections $\qsetVone$, $\qsetVtwo$, $\qsetCone$ and $\qsetCtwo$ of $\qset$ are indicated.
The neighbors $\qnbhd(\qset)$ within $V_1 \times C_2$ are depicted using square nodes; the projections $\nbhd(\qsetVtwo)$ and $\nbhd(\qsetCone)$ are indicated.
For a specific $v_2 \in V_2$, the neighbors $\qnbhd(\qset(v_2) \times \{v_2\}) = \nbhd(\qset(v_2)) \times \{v_2\}$ are also depicted along the dashed line.}
\label{fig:projections}
\end{figure}
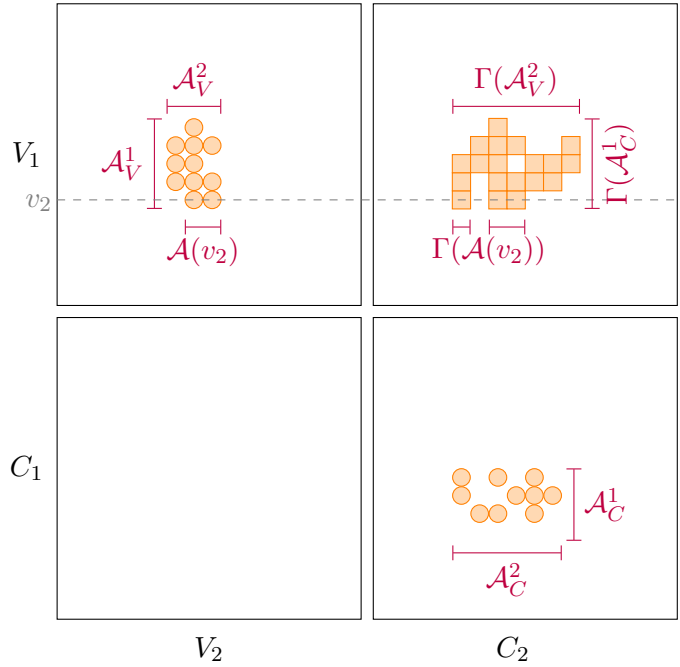

\subsection{Properties of reduced sets}
\label{subsec:properties-reduced-sets}

Recall from Section~\ref{sec:techoverview} that we are allowed to ``toggle'' elements of $\supp(z)$ for $z \in \gens$ without changing the codeword (that is, as above, we are working modulo $\ccode_Z^\perp$). 
To this end, we define a \emph{reduced} error as one of minimum weight modulo this toggling.  Formally, we have the following definition:
\begin{definition}[Reduced representation]
    Given $\err \subset \qubits$, let $[\err]$ be the equivalence class of $\err$ defined by
    \begin{equation*}
        [\err] = \{ \err \oplus \supp | \supp \in \ccode_Z^{\perp} \}~.
    \end{equation*}
    The \emph{reduced representation} $\red{\err} \in [\err]$ is the smallest element in $[\err]$.
    If there are several elements with the minimum weight, we pick one arbitrarily.
\end{definition}

Similarly, we can define a locally reduced set in the support of a generator.
\begin{definition}[Locally Reduced set]
 Let $\gen$ be a generator and let $\supp(\gen)$ be its support. We say that a set $\add \subseteq \supp(\gen)$ is \emph{locally reduced} if
\begin{align}
    |\add_V| + |\add_C| \leq \left(\Delta_V - |\add_V|\right) + \left(\Delta_C - |\add_C| \right)~.
\end{align}   
\end{definition}
If $\add$ is locally reduced then we can bound the size of $|\add_V| \cdot |\add_C|$ via the following lemma, where we recall the notation $\|\add\|$ from Definition~\ref{def:wtd_norm}.

\begin{lemma}
\label{lem:reduced-nbhd}
    If $\add \subseteq \supp(\gen)$ is locally reduced,  then 
    \begin{align*}
        |\add_V| \cdot |\add_C| \leq \frac{1}{4} \Delta \lVert \add \rVert~.
    \end{align*}
\end{lemma}
\begin{proof}
    We can write the claim as 
    \begin{align}
        |\add_V| \left( \frac{\Delta_V}{2} - |\add_C| \right) + |\add_C| \left( \frac{\Delta_C}{2} - |\add_V| \right) \geq 0~.
    \end{align}
    Consider the LHS, we can write it as 
    \begin{align}
        \min(|\add_V|, |\add_C|) \cdot \left[ \left( \frac{\Delta_V}{2} - |\add_C| \right) + \left(\frac{\Delta_C}{2} - |\add_V| \right) \right] \geq 0
    \end{align}
    By assumption, $\add$ is locally reduced and this implies
    \begin{align}
        |\add_V| + |\add_C| \leq \frac{1}{2} \left( \Delta_V + \Delta_C \right)~.
    \end{align}
    Furthermore, $\min(|\add_V|,|\add_C|)$ is non-negative.
\end{proof}

\section{Viderman's algorithm for Hypergraph Product Codes}\label{sec:HPCssflip}

In this section, we present a quantum version of Viderman's algorithm for hypergraph product codes, and prove that it is correct.
The objective of this section is to prove Theorem~\ref{thm:main}.

Our algorithm, $\ssfind$, is presented in Section~\ref{subsec:statement-ssfind}.  Given a set of $Z$-type errors $\err \subseteq \qubits$, $\ssfind$ produces an envelope $\env$ such that it contains the reduced version of the error $\err$.
In Section~\ref{subsec:coverage}, we prove that $\env$ contains the reduced version of $\err$.
In Section~\ref{subsec:bounded-L}, we prove that when $\err$ is small enough, $\env$ does not grow too much.

For the entirety of this section, we assume we have a quantum expander code that is constructed from a graph $\cG$, that is the graph product of two copies of a graph $G$ that is a $(\Delta_V,\Delta_C)$-biregular, $(\alpha_V,\epsilon_V,\alpha_C,\epsilon_C)$-bidirectional expander as described in Section~\ref{subsec:background-quantum-expander}.
For clarity, we denote the two copies as $G_1=(V_1\cup C_1,E_{1}),G_2 =(V_2\cup C_2,E_{2})$.
 
\subsection[ssfind]{$\ssfind$}
\label{subsec:statement-ssfind}

We begin by recalling from Section~\ref{sec:techoverview} why $\find$ (Algorithm~\ref{alg:find-gentle}, the classical version of Viderman's algorithm) does not extend directly to the quantum setting.
In each iteration, the classical algorithm $\find$ adds a single bit to the envelope $\env$ if it has a lot of overlap with the set of untrustworthy checks $\sus$.
To be precise, the $\find$ algorithm adds a bit $v$  to the envelope if at least $(1-2\epsilon_V)\cdot \Delta_V$ of the neighbors of $v$ are in $\sus$.
Suppose we encounter the example in Figure \ref{fig:viderman-fail} (a reprise of Figure~\ref{fig:viderman-fail-intro}).
There are errors only within the support of a single generator $\gen \sim (c,v) \in \gens$, and these errors are in both $V_1\times V_2$ qubits and in $C_1\times C_2$ qubits.
The shaded rectangles indicate $\unsat$.
Parity checks that are adjacent to two errors are satisfied, and therefore no error qubit has $(1-2\epsilon_V)\Delta_V$ unsatisfied parity checks adjacent to it. If we reduce the threshold of the $\find$ algorithm to cover errors in these cases, we would have to reduce it from $(1-2\epsilon_V)\Delta_V$ to $1/2\Delta_V$, 

and then the envelope might end up growing uncontrollably and cover the entire set of qubits.

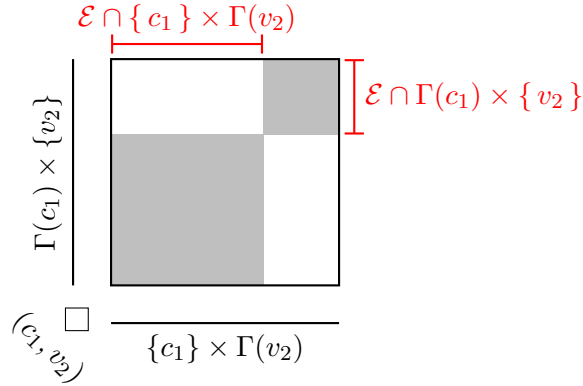
\begin{figure}[h]
    \centering
    \begin{tikzpicture}
        \fill[gray!50] (0,0) rectangle (2,2);
        \fill[gray!50] (2,2) rectangle (3,3);
        \draw[thick] (0,0) rectangle (3,3);

        \draw[thick] (0,-0.5)--node[below] {$\{c_1\} \times \nbhd(v_2)$}(3,-0.5);
        \draw[thick] (-0.5,0)--node[above,rotate=90] {$\nbhd(c_1) \times \{v_2\}$}(-0.5,3);
        \draw[|-|,thick,red] (0,3.2)--node[above] {$\err \cap \set{c_1}\times\nbhd(v_2)$}(2,3.2);
        \draw[|-|,thick,red] (3.2,3) -- node[right]{$\err\cap \nbhd(c_1)\times\set{v_2}$} (3.2,2);
        \draw (-0.3,-0.3) rectangle (-0.6,-0.6) node[below,rotate=-45] {$(c_1,v_2)$};
    \end{tikzpicture}
    \caption{The neighborhood of a generator $(c_1,v_2)$. Error is marked in red and unsatisfied parity checks are shaded.}
    \label{fig:viderman-fail}
\end{figure}

As discussed in Section~\ref{sec:techoverview},
we overcome this problem by taking into consideration neighborhoods of \emph{sets} of qubits rather than the neighborhood of one qubit at a time.
Let $\supp(\gen)$ be the support of $\gen \in \gens$.
In Definition~\ref{def:score}, we define a score over sets $\add \subset \supp(\gen)$ that, as discussed in Section~\ref{sec:techoverview}, accounts for the fact that parity checks that witness \emph{both} VV errors and CC errors can be satisfied; that is, ``cancellations'' can occur.
The score is then used in the algorithm $\ssfind$ (Algorithm~\ref{alg:ssfind}) to identify ``suspicious'' qubits. 

\begin{definition}
\label{def:reduced-errs}
For each generator $\gen \in \gens$, we define the collection of locally reduced sets in the neighborhood of $\gen$ to be
\begin{align*}
    \minsets(\gen) &= \left\{\add \subseteq \supp(\gen)  \; \middle| \;  \; |\add_V| + |\add_C| \leq \frac{1}{2}(\Delta_V + \Delta_C)\right\}~.
\end{align*}
Let $\minsets(\gens) = \bigcup_{\gen \in \gens} \minsets(\gen)$ be the union over locally reduced sets across all generators.
\end{definition}

The $\ssfind$ algorithm will maintain an envelope $\env$ of qubits it believes are suspicious, and a set of parity checks $\sus$ that are also suspicious.  The algorithm labels a parity check as suspicious whenever it touches a suspicious qubit, and thus we always have $\qnbhd(\env) \subset \sus$.
In each iteration, $\ssfind$ searches over $\add \in \minsets(\gens)$ and evaluates their \emph{score}, which we formally define below.  If the score of a set $\add$ is small, it is considered suspicious. See Section~\ref{sec:techoverview} for a discussion of the intuition behind this score function.

\begin{definition}[Score]
\label{def:score}
    Let $\sus \subseteq \cX$ be a set of parity checks and let $\sus^c$ be its complement.
    For $\add \in \minsets(\gens)$, we define the \emph{score} of $\add$ as
    \begin{align*}
        \score(\add) &= \frac{\lvert \uqnbhd(\add) \cap \sus^c \rvert}{\Delta\lVert \add \rVert}~.
    \end{align*}
\end{definition}
We note that the score of a set $\add$ depends on the set of suspected parity checks $\sus$.
However, we suppress this dependence in the notation for readability.
While the set $\sus$ changes over the course of the algorithm, the score of subsets of qubits $\add$ within the support of a generator $\gen$ only changes if the set of parity checks in the local view are added to $\sus$.

To illustrate, we use an example and corresponding Figure~\ref{fig:view-def}.
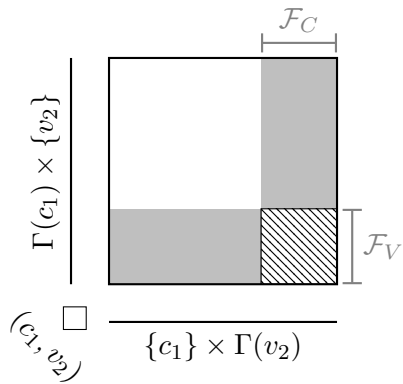
\begin{figure}[h]
    \centering
    \begin{tikzpicture}
   
        \fill[gray!50] (2,1) rectangle (3,3);
        \draw[thick,gray,|-|] (3.2,0)--node[right] {$\add_V$} (3.2,1);
        \fill[gray!50] (0,0) rectangle (2,1);
        \draw[thick,gray,|-|] (2,3.2)--node[above] {$\add_C$} (3,3.2);
        \draw[pattern=north west lines] (2,0) rectangle (3,1);
        \draw[thick] (0,0) rectangle (3,3);
             
        \draw[thick] (0,-0.5)--node[below] {$\{c_1\} \times \nbhd(v_2)$}(3,-0.5);
        \draw[thick] (-0.5,0)--node[above,rotate=90] {$\nbhd(c_1) \times \{v_2\}$}(-0.5,3);
        \draw (-0.3,-0.3) rectangle (-0.6,-0.6) node[below,rotate=-45] {$(c_1,v_2)$};
    \end{tikzpicture}
    \caption{The local view of a generator $\gen \sim (c_1,v_2)$.
    $\add_V$ and $\add_C$ form the reduced set $\add$.
    The gray portions correspond to $\uqnbhd(\add)$.
    The hatched portion corresponds to $\mqnbhd(\add)$.}
    \label{fig:view-def}
\end{figure}
First, $\add$ is chosen such that it is locally reduced.
If $\add$ is the only error in the code, the syndromes would correspond to the union of the gray rectangles.

In Algorithm~\ref{alg:ssfind}, we present $\ssfind$.
Let $\err \subseteq \qubits$ be the error and $\unsat$ denote the set of unsatisfied parity checks.
Recall that we denote $\epsilon = \max(\epsilon_V,\epsilon_C)$.
The algorithm $\ssfind$ takes as input $\unsat$.
It outputs an envelope $\env\subset\qubits$ that are suspected errors.

\begin{algorithm}[H]
    \begin{algorithmic}[0]
        \State \textbf{Input:} $\unsat \subseteq \pcs$
        \State \textbf{Output:} $\env \subseteq \qubits$.
    \end{algorithmic}
    \begin{algorithmic}[1]
        \State $\env \leftarrow \emptyset$
        \State $\sus \leftarrow \unsat$
        \State $\add \leftarrow \minsets(\gens)$
        \While{$\exists \add \in \minsets(\gens)$ such that $\score(\add) \leq 2\epsilon $}
            \State $\env \leftarrow \env \union \add$. \label{step:add}
            \State $\sus \leftarrow \sus \union \qnbhd(\add)$
            \For{$\add' \in \minsets(\gens)$ such that $\add' \cap \add\neq\emptyset$}
                \State $\minsets(\gens) \leftarrow \minsets(\gens) \setminus \{\add'\}$
            \EndFor
        \EndWhile
        \State \Return $\env$.
    \end{algorithmic}
    \caption{$\ssfind$}
    \label{alg:ssfind}
\end{algorithm}
We remark that there exists an implementation of this algorithm that runs in time $\Theta(|\err|)$, see Section~\ref{sec:time} for more details.

\subsection{Coverage}
\label{subsec:coverage}

In this section, we prove that $\ssfind$ produces an envelope $\env$ that covers the reduced error $\err$.
To this end, we shall assume that it does not and arrive at a contradiction.

Without loss of generality, assume the error is already in its reduced form, i.e.\ $\red{\err} = \err$. 
Let $\res = \env \setminus \err$ be the residual errors that are not covered by the envelope when $\ssfind$ terminates.
For the sake of contradiction, we assume $\res$ is not empty.

We begin with some high-level intuition and then build on it to arrive at the proof.
The idea is to consider the set $\uqnbhd(\res)$, the unique neighbors of the residual errors.
The key observation is that elements of the set $\uqnbhd(\res)$ must be in $\sus$.

\begin{lemma}
    \label{lem:unbhd-res-to-sus}
    Let $\pc \in \pcs$ such that $\pc \sim (\nu_1,\zeta_2)$ is an element of $\uqnbhd(\res)$.
    Then $\pc \in \sus$.
\end{lemma}
\begin{proof}
    If $\pc \in \uqnbhd(\res)$, then by definition, it is incident to exactly one qubit in $\res$.
    This leaves two choices.
    If this parity check is not connected to an element from $\env$, then it must be connected only to a single error (as $\err = (\err\cap\env)\sqcup \res$), and therefore must be in $\unsat \subseteq \sus$ .
    On the other hand, if it \emph{is} connected to an element from $\env$, then it must be in $\qnbhd(\env)\subset\sus$.
\end{proof}
Intuitively, if $\uqnbhd(\res)$ is large, then there are elements in $\res$ such that their neighborhood has a large overlap with $\sus$.
In particular, we will show that $\uqnbhd(\res)$ has a large overlap with the neighborhood of a generator $\gen \in \gens$.
This will imply that $\gen$ contains a portion of $\res$ with a large score.
In turn, this will lead to a contradiction of the assumption that the $\ssfind$ algorithm terminated.

The following lemma closely resembles Lemma 7 in \cite{leverrier2015quantum}.

\begin{lemma}
\label{lem:coverage}
   The output $\env$ of $\ssfind$ covers the set $\err$, i.e.\ $\err \subseteq \env$ if $|\err_V| \leq  \alpha_V n$ and $|\err_C| \leq \alpha_C m$.
\end{lemma}
\begin{proof}
    Assume without loss of generality that $\res_V\neq\emptyset$. If $\res_V=\emptyset$, we can replace $V$ with $C$ and rows with columns. Let $\resVtwo\subset V_2$ be the projection of $\res$, as defined in Definition~\ref{def:proj}:
    \[\resVtwo = \set{v_2\in V_2 | \exists(\nu_1,v_2) \in\res }~.\]
    Notice that the set $\resVtwo\subset V_2$, which is a set in the graph $G_2$.
    
    As $|\resVtwo| \leq |\res| \leq |\err| \leq \min(\alpha_V n, \alpha_C m)$ and the graph $G_2$ is an $(\alpha_V,\epsilon_V)$ left vertex-expander,
    \begin{align}
        |\unbhd(\resVtwo)| \geq (1-2\epsilon_V)\Delta_V|\resVtwo|~.
    \end{align}
    This then implies that there exists a node $v_2 \in \resVtwo$ such that it has a large overlap with the unique neighborhood:
    \begin{align}
        |\unbhd(\resVtwo) \cap \nbhd(v_2)| \geq (1-2\epsilon_V)\Delta_V~.
    \end{align}
    We fix this node $v_2$ for the rest of the proof.

    Let $\resCtwo\subset C_2$ be as in Definition~\ref{def:proj}, i.e. $\resCtwo = \set{\zeta_2\in C_2 | \exists (c_1,\zeta_2)\in \res} $. We divide into two cases.
    
    \textbf{Classical-like case:}
    First, consider the simpler case where
    \begin{align}
    \label{eq:case1}
        \unbhd(\resVtwo) \cap \nbhd(v_2) \cap \resCtwo = \emptyset~.
    \end{align}
    This is the setting where most parity checks adjacent to a qubit of the form $(\nu_1,v_2) \in \res_V$ are not simultaneously adjacent to CC qubits.
    In this sense, the setting is like a classical decoding problem where we do not consider interference from CC qubits.
    
    Let $\nu_1 \in V_1$ be an element such that $(\nu_1,v_2) \in \res_V$, there must be one since $v_2\in\resVtwo$.
    Assuming Eq.\ \eqref{eq:case1}, it follows that for all $c_1 \in \nbhd(\nu_1)$,
    \begin{align}\label{eq:no-err}
        \unbhd(\resVtwo) \cap \nbhd(v_2) \cap \rescone = \emptyset,
    \end{align}
    where $\rescone$ is the projection of $\res$ on the row containing $c_1$, see Definition~\ref{def:proj}.
    This is because, by definition, $\rescone \subseteq \resCtwo$. Fix such $c_1\in\nbhd(\nu_1) $ and let $z$ be $(c_1,v_2)$.
    Equation (\ref{eq:no-err}) means that the neighborhood of the generator $\gen \in \gens$ is simple, as $\set{c_1}\times \unbhd(\resVtwo) \cap \nbhd(v_2) $ contains no elements from $\res$.
    See Figure~\ref{fig:localviewcase1}.

    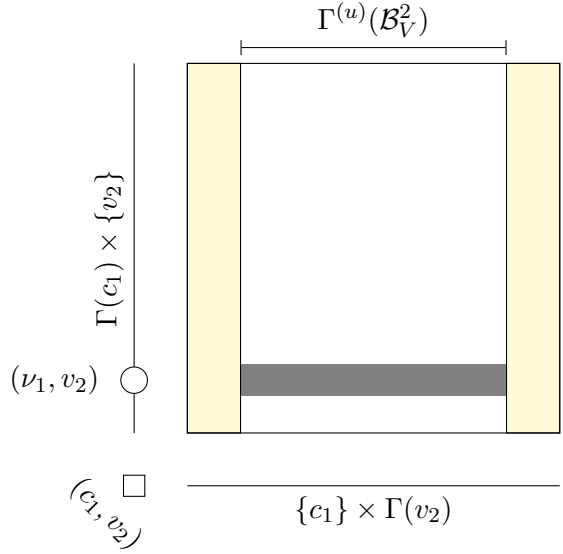
\begin{figure}[h]
       \captionsetup{singlelinecheck=off}
    \centering
    \begin{tikzpicture}[scale=0.7]
        \fill[yellow!20!white] (0,0) rectangle (1,7);
        \fill[yellow!20!white] (6,0) rectangle (7,7);
        \draw (0,0) rectangle (7,7);
        \draw (-0.8,-0.8) rectangle (-1.2,-1.2) node[below,rotate=-45] {$(c_1,v_2)$};
        \draw (-1,0)--node[above,rotate=90] {$\nbhd(c_1) \times \{v_2\}$} (-1,7) (0,-1)-- node[below] {$\{c_1\} \times \nbhd(v_2)$}(7,-1);
        \fill[gray] (1,1.3) rectangle (6,0.7);
        \draw[fill=white] (-1,1) circle (0.25);
        \node at (-2.5,1) {$(\nu_1,v_2)$};
        \draw[|-|] (1,7.3) -- node[above] {$\unbhd(\resVtwo)$} (6,7.3);
        \draw (0,0) rectangle (1,7);
        \draw (6,0) rectangle (7,7);
    \end{tikzpicture}
    \caption[]{We find a generator $\gen \sim (c_1,v_2)$ such that $(\nu_1,v_2)\in\err_V$ is in its support and $\left[\set{c_1}\times \unbhd(\resVtwo) \cap \nbhd(v_2)\right] \cap\res = \emptyset $.
    We prove that the gray line is in $\sus$, i.e.\ $\set{\nu_1}\times(\unbhd(\errVtwo) \cap \nbhd(v_2))\subset\sus$.
    We ignore the parity checks corresponding to the yellow portions.
    We write $\unbhd(\resVtwo)$ for $\unbhd(\resVtwo) \cap \nbhd(v_2)$ for readability of the diagram.
    }
    \label{fig:localviewcase1}
    \end{figure}
    Let $\pc = (\nu_1,\zeta_2)$ be a parity check, where $\zeta_2\in \unbhd(\resVtwo) \cap \nbhd(v_2)$. We claim that $\pc$ sees only a single qubit in $\res$, the qubit $(\nu_1,v_2)$. The parity check $\pc$ cannot be adjacent to other elements of $\res_V$ because $\zeta_2 \in \unbhd(\errVtwo) \cap \nbhd(v_2)$. Furthermore, it cannot be adjacent to other elements of $\res_C$ because of (\ref{eq:no-err}).
    
    From Lemma~\ref{lem:unbhd-res-to-sus}, $\pc$ is an element of $\sus$.
    This implies that $|\uqnbhd(\{\nu_1,v_2\}) \cap \sus|$ is lower bounded:
    \begin{align}
        |\uqnbhd(\{\nu_1,v_2\}) \cap \sus|
        &\geq \left\lvert \{\nu_1\} \times \unbhd(\resVtwo) \cap \nbhd(v_2) \right\rvert\\
        &\geq (1-2\epsilon_V)\Delta_V~.
    \end{align}
    In other words, the score is upper bounded:
    \begin{align}
        \score(\{\nu_1,v_2\}) = \frac{|\uqnbhd(\{\nu_1,v_2\}) \cap \sus^c|}{\Delta_V} \leq 2\epsilon~.
    \end{align}
    This means that $\add = \{(\nu_1,v_2)\}$ would have been added to $\env$ and the algorithm could not have terminated.

    \textbf{Quantum-like case:}
    Suppose that the VV and CC portions of $\res$ interfere, i.e.\
    \begin{align}
        \indCtwo := \unbhd(\resVtwo) \cap \nbhd(v_2) \cap \resCtwo \neq \emptyset~.
    \end{align}
    We shall refer to $\indCtwo\subset C_2$ as the induced set.
    If $\indCtwo$ is not empty, then we can construct the set $\indCone$
    \begin{align}
        \indCone = \set{c_1\in C_1 | \exists \zeta_2 \in \indCtwo \text{ such that } (c_1,\zeta_2) \in \res_C}~.
    \end{align}
    The induced set, being a subset of $\resCone$, must be small
    \begin{align}
        |\indCone| \leq |\resCone| \leq |\res_C| \leq |\err_C| \leq \min(\alpha_V n, \alpha_C m)~.
    \end{align}
    Therefore, it has a large unique neighborhood in $G_1$, $|\unbhd(\indCone)| \geq (1-2\epsilon_C)\Delta_C|\indCone|$.
    Consequently, there must be some $c_1 \in \indCone$ such that
    \begin{align}
        |\unbhd(\indCone) \cap \nbhd(c_1)| \geq (1-2\epsilon_C)\Delta_C~.
    \end{align}
    We identify the generator $\gen= (c_1,v_2)$, and show that this generator has a large overlap with $\unbhd(\res)$.
    See Figure~\ref{fig:localviewcase2}.

    \begin{figure}[h]
    \captionsetup{singlelinecheck=off}
    \centering
    \begin{tikzpicture}[scale=0.7]
        \fill[yellow!20!white] (0,0) rectangle (1,7);
        \fill[yellow!20!white] (6,0) rectangle (7,7);
        \fill[yellow!20!white] (1,0) rectangle (6,1);
        \fill[yellow!20!white] (1,6) rectangle (6,7);
        \draw (0,0) rectangle (7,7);
        \draw (-0.8,-0.8) rectangle (-1.2,-1.2) node[below,rotate=-45] {$(c_1,v_2)$};
        \draw (-1,0)--(-1,7) (0,-1)--(7,-1);
        \fill[gray] (2,1) rectangle (6,2);
        \draw[|-|] (1.0,-1.3)--node[below] { \large $\add_C$} (2.1,-1.3);
        \fill[gray] (1,2) rectangle (2,6);
        \fill[pattern=north west lines](1,2) rectangle (2,1);
        \draw[|-|] (-1.3,1.0) --node[above,rotate=90] { \large $\add_V$} (-1.3,2.1);
        \draw[black] (1,0)--(1,7);
        \draw[black] (6,0)--(6,7);
        \draw[|-|] (7.3,0.9) --node[below,rotate=90] {$\unbhd(\indCone)$} (7.3,6.1);
        \draw[black] (0,1)--(7,1);
        \draw[black] (0,6)--(7,6);
        \draw[|-|] (0.9,7.3) --node[above] {$\unbhd(\resVtwo)$} (6.1,7.3);
        \draw (0,0) rectangle (1,7);
        \draw (6,0) rectangle (7,7);
        \draw (0,0) rectangle (7,1);
        \draw (0,6) rectangle (7,7);
    \end{tikzpicture}
    \caption[]{$\add_V,\add_V\subset\err$, and the parity checks in $\unq = \unbhd(\resVtwo)\times \unbhd(\indCone)$ only interact with errors in $\supp(\gen)$, the support of $\gen$.
    We show that the parity checks in gray, $\uqnbhd(\add)\cap\unq$, are in $\sus$.
    We ignore the parity checks corresponding to the yellow portions.
    We write $\unbhd(\resVtwo)$ for $\unbhd(\resVtwo) \cap \nbhd(v_2)$ and $\unbhd(\indCone)$ for $\unbhd(\indCone) \cap \nbhd(c_1)$ for readibility of the diagram.
}
    \label{fig:localviewcase2}
    \end{figure}
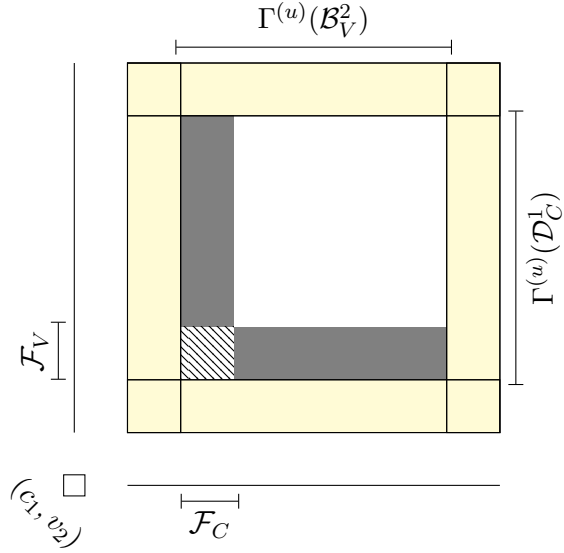

    Consider the set $\add = \add_V \sqcup \add_C$, where
    \begin{align}
     \add_V := \left(\resvtwo \cap \unbhd(\resCone) \cap \nbhd(c_1)\right) \times \{v_2\}~, \qquad
     \add_C := \{c_1\} \times \left(\rescone \cap \unbhd(\resVtwo) \cap \nbhd(v_2) \right)~.
    \end{align}
    Note that all elements in $\add$ are errors, from the definition of $\resvtwo,\rescone$, see Definition~\ref{def:proj}.

    We claim that the set $\add_C$ is not empty.
    By construction, $c_1 \in \indCone$, and therefore, for some $\zeta_2 \in \indCtwo$, we must have $(c_1,\zeta_2)\in\res_C$. By the definition of $ \indCtwo$, this means that $\zeta_2\in \unbhd(\resVtwo)\cap \nbhd(v_2)$.

    As in the classical-like case, we want to show that most neighbors of $\add$ are unique neighbors of $\res$ and therefore in $\sus$. 
    We study this generator using the set of parity checks $\unq = \unbhd(\indCone) \times \unbhd(\resVtwo)$, the product of unique neighborhoods. 
    \begin{claim}
        Let $\pc \sim (\nu_1,\zeta_2)$ be a parity check in $\uqnbhd(\add)$ within the set $\unq$,
        then $\pc \in \sus$.
    \end{claim}
    \begin{claimproof}
        As $\pc \in \uqnbhd(\add)$, it can only be adjacent to $\add_V$ or $\add_C$ but not to both.
        Without loss of generality, suppose it is adjacent to $\add_C$.
        This implies that $\zeta_2 \in \unbhd(\resVtwo) \cap \nbhd(v_2)$ and $\nu_1 \in \unbhd(\indCone) \cap \nbhd(c_1) \setminus \resvtwo$.
        
        The parity check $\chi$ cannot be adjacent to other elements of $\res_C$.
        This is because $\nu_1 \in \unbhd(\indCone) \cap \nbhd(c_1)$.
        If, for the sake of contradiction, there was a $c_1' \in C_1$ such that $c_1' \in \nbhd(\nu_1)$ and $(c_1',\zeta_2) \in \res_C$.
        Then $c_1'$ would also be in $\indCone$ and $\nu_1$ would no longer be a unique neighbor of $\indCone$.

        The parity check $\chi$ cannot be adjacent to other elements of $\res_V$. 
        This is because $\zeta_2 \in \unbhd(\resVtwo) \cap \nbhd(v_2)$, and $v_2\in \resVtwo$.

        Therefore, $\pc \in \unbhd(\res)$ and by Lemma~\ref{lem:unbhd-res-to-sus}, $\pc \in \sus$.
    \end{claimproof}
    This claim implies that
    \begin{align}
    \label{eq:overlap-unbhd}
        \uqnbhd(\add) \cap \unq\subset\sus~.
    \end{align}
    As the unique neighborhood has a large footprint within the generator $\gen$, only an $\epsilon$ fraction of parity checks remain outside the overlap with $\unbhd(\indCone)\times \unbhd(\resVtwo)$, i.e.\
    \begin{align}
    \label{eq:setminus-unbhd}
        \left| \uqnbhd(\add) \setminus \unq \right| \leq
        2\epsilon_C \Delta_C |\add_C| + 2\epsilon_V\Delta_V |\add_V|~.
    \end{align}

    Together, Eq.\ \eqref{eq:overlap-unbhd} and Eq.\ \eqref{eq:setminus-unbhd} imply
    \begin{align}
        |\uqnbhd(\add) \cap \sus| \geq |\uqnbhd(\add) \cap \unq|\geq |\uqnbhd(\add)| - 2\epsilon_V\Delta_V|\add_V| - 2\epsilon_C\Delta_C|\add_V|~,
    \end{align}
    which implies that
    \begin{align}
        \score(\add) = \frac{|\uqnbhd(\add) \cap \sus^c|}{\Delta \lVert \add \rVert} \leq 2\epsilon~.
    \end{align}
    This completes the proof.
\end{proof}

\subsection{Bounding the size of the envelope}\label{subsec:bounded-L}
In the previous section, we have shown that the envelope $\env$ will cover the error.
However, this does not guarantee that the envelope will stay small: if it ends up, say, covering cover \emph{all} the qubits, then our algorithm will not be very useful.
In this section, we show that the envelope $\env$ stops growing such that its total size is bounded.

Following the high-level outline of the proof in the classical case (see Appendix~\ref{sec:classic-code}),
we bound the size of $\env$ by lower and upper bounding $\norm{\qnbhd(\env)}$.
The lower bound comes from the graph expansion.
For the upper bound, we note that in each iteration, the set $\add$ that is added to $\env$ must have a large intersection with $\sus$.
The set $\sus$ contains unsatisfied parity checks and $\qnbhd(\env)$.
Once all the unsatisfied parity checks are covered by $\env$, every new set $\add$ we add to $\env$ has a large overlap with $\env$.
Therefore, this reduces the expansion of the set $\qnbhd(\env)$ (or equivalently, the ratio $|\qnbhd(\env)|/|\env|$ can only reduce once all $\unsat$ have been added).

\subsubsection{Lower bounding the expansion}
The lower bound on $\norm{\qnbhd(\env)}$ comes from the expansion of the base graph $G$.
\begin{lemma}
    \label{lem:max-expanding-size}
    Let $\qset\subset \qubits$ be a set such that $|\qset_V| \leq \alpha_Vn$ and $|\qset_C| \leq \alpha_C m$.
    Then
    \begin{align}
        \norm{\qnbhd(\qset)}\geq \frac{1}{2}(1-\epsilon)\Delta\lVert \qset \rVert~.
    \end{align}
\end{lemma}

\begin{proof}
    The proof uses the fact that the input graph $G$ is an expander and that $\cG$ is the product of $G$.
    Every row indexed by $\nu_1$ in $\cG_\pcs$ is a copy of $G$.
    Therefore, if we restrict $\cG_\pcs$ to a row, we can inherit the expansion of $G$. 
    Similarly, every column indexed by $\zeta_2$ in $\cG$ is also a copy of $G$.
    
    Recall that for every $\nu_1\in V_1$,  $\qset(\nu_1)$ is the projection of $\qset$ on the line indexed by $\nu_1$:
    \[\qsetvone = \left\{ v_2\in V_2 | (\nu_1,v_2)\in \qset \right\}~.\]
    
    By assumption, $|\qset_V|\leq \alpha_V n$ and $|\qset_C| \leq \alpha_C m$.
    Therefore, each row $\qsetvone$ and each column $\qsetzetatwo$ is expanding.
    We can write
    \begin{align}
        |\qnbhd(\qset_V)|
        &=\sum_{\nu_1}|\{\nu_1\} \times \nbhd(\qsetvone)|\\
        &\geq \sum_{\nu_1} (1-\epsilon_V)\Delta_V |\qsetvone|\\
        &\geq (1-\epsilon_V)\Delta_V|\qset_V|~.
    \end{align}
    An identical argument works for $\qset_C$: $|\qnbhd(\qset_C)| \geq (1-\epsilon_C)\Delta_C|\qset_C|$.
    
    There might be an overlap between $\qnbhd(\qset_V)$ and $\qnbhd(\qset_C)$, therefore we lower bound $\qnbhd(\qset)$ by its maximum value,
    \begin{align}
      \norm{\qnbhd(\qset)}\geq& \max\set{|\qnbhd(\qset_V)|,|\qnbhd(\qset_C)| }
      \\\geq&\max\set{(1-\epsilon_V)\Delta_V|\qset_V|,(1-\epsilon_C)\Delta_C|\qset_C| }~.
    \end{align}
    By definition,  $\Delta\lVert \qset\rVert =  \Delta_V |\qset_V|+\Delta_C|\qset_C|$.
    Therefore $\max( \Delta_V |\qset_V|,\Delta_C|\qset_C|)\geq\frac{1}{2}\Delta\lVert \qset\rVert $.
    Using the fact that $\epsilon =\max(\epsilon_V,\epsilon_C)$, we finish the proof. 
\end{proof}

    As an aside, we give an example of a set of qubits $\qset$ such that $\norm{\qnbhd(\qset)} \leq \frac{1}{2}\Delta \lVert \qset\rVert$, showing that the bound in the lemma above is tight up to the $\epsilon$ factor.
    In particular, the factor of $1/2$ is unavoidable.
\begin{example}[Minimally expanding set]
    We define a set $\qset \subset \qubits$ by first choosing arbitrary sets $C_1' \subset C_1$ and $V_2' \subseteq V_2$ such that $|\qsetCone| \leq \alpha_C m$ and $|\qsetVtwo| \leq \alpha_V n$.
    
    Define $\qset =\qset_V\cup \qset_C $ when $\qset_V = \nbhd(C_1') \times V_2'$ and $\qset_C = C_1' \times \nbhd(V_2')$.
    Then we have that 
    \begin{align}
        \qnbhd(\qset_V) = \nbhd(C_1') \times \nbhd(V_2') = \qnbhd(\qset_C)~.
    \end{align}

    This means that
    $\norm{\qnbhd(\qset)} \leq \Delta_V\norm{\qset_V}$ and also $\norm{\qnbhd(\qset)} \leq \Delta_C\norm{\qset_C}$.
    Using the fact that $\Delta \lVert \qset\rVert = \Delta_V\norm{\qset_V}+\Delta_C\norm{\qset_C}$ we find $\norm{\qnbhd(\qset)}\leq\frac{1}{2}\Delta\lVert \qset\rVert$.
\end{example}
We remark that the algorithm might output an envelope $\env$ with expansion as in the example above. suppose that the error $\err$ is exactly half of the support of the generator.
Then when the algorithm terminates, the envelope will contain its entire support, as in the example.

\subsubsection{Upper bounding the expansion}
\begin{lemma}
\label{lem:upp-L-bound}
    Let $\env\subset \qubits$ be an envelope during the run of Algorithm \ref{alg:ssfind}. 
    Then
    \begin{align*}
        \norm{\qnbhd(\env)}\leq \norm{\unsat} + \left( \frac{1}{4} + 2\epsilon\right)\Delta\lVert \env \rVert
    \end{align*}
\end{lemma}
\begin{proof}
Let $\env$ be an envelope, and suppose that in some iteration, the $\ssfind$ algorithm adds $\add$ to $\env$.

By definition, we only add sets $\add$ if $\score(\add) \leq 2\epsilon$.
Noting that $|\qnbhd(\add) \cap \sus| \geq |\uqnbhd(\add) \cap \sus|$ and recalling the definition of $\score$, we see that this is satisfied provided that 
    \begin{align}\label{eq:F-cond}
        \norm{\qnbhd(\add)\cap \sus}\geq |\uqnbhd(\add)| - 2\epsilon\Delta\lVert\add\rVert~,
    \end{align}
The set $\sus$ is the set of parity checks defined by $\sus=\qnbhd(\env)\cup \unsat$.
It will be convenient to write $\sus$ as a disjoint union,
\begin{align}
    \sus = \qnbhd(\env) \sqcup (\unsat \setminus \qnbhd(\env))~.
\end{align}
Using this partition of $\sus$ we can write
\begin{align}
\label{eq:partition-cap}
    |\qnbhd(\add) \cap \sus| = |\qnbhd(\add) \cap \qnbhd(\env)| + |\qnbhd(\add) \cap (\unsat \setminus \qnbhd(\env))|~.
\end{align}
We can use Eq.\ \eqref{eq:partition-cap} to bound the number of new neighbors $\add$ adds to $\env$,
\begin{align}
    \lvert \qnbhd(\add) \setminus \qnbhd(\env) \rvert
    &= \lvert \qnbhd(\add) \rvert - \lvert \qnbhd(\add) \cap \qnbhd(\env) \rvert\\
    &= \lvert \qnbhd(\add) \rvert - \lvert \qnbhd(\add) \cap \sus\rvert + |\qnbhd(\add) \cap (\unsat \setminus \qnbhd(\env))|
    \\&\leq \norm{\mqnbhd(\add)} + 2\epsilon\Delta\lVert \add \rVert + \norm{\qnbhd(\add) \cap (\unsat \setminus \qnbhd(\env))}~.
\end{align}
The last inequality uses (\ref{eq:F-cond}) to bound $ \lvert\qnbhd(\add) \rvert - \lvert \qnbhd(\add) \cap \sus\rvert$.
As Algorithm \ref{alg:ssfind} only adds reduced sets $\add_j$, i.e. sets $\add_j \subset \supp(\gen)$ for some generator $\gen$, such that $\norm{\add_i}\leq\frac{\Delta_C+\Delta_V}{2}$.
We can therefore use Lemma~\ref{lem:reduced-nbhd} to bound $|\mqnbhd(\add)|$:
\begin{align}
    |\qnbhd(\add) \setminus \qnbhd(\env)| \leq \left(\frac{1}{4}+2\epsilon\right)\lVert\add\rVert\Delta + |\qnbhd(\add) \cap (\unsat \setminus \qnbhd(\env))|~.
\end{align}

Suppose the algorithm works for $i$ rounds, and let $\env_i$ be the envelope in the $i$\textsuperscript{th} iteration.
Denote by $\add_i$ the set that we add in the $i$\textsuperscript{th} round, i.e.\ $\env_i = \env_{i-1} \union \add_i$, when $\add_i$ satisfies the condition in the claim. Then we have
\begin{align}
    \norm{\qnbhd(\env_i)} \leq& \sum_{j\leq i} \norm{\qnbhd(\add_j \setminus \env_{j-1})}\\
    \leq& \sum_{j \leq i} \left[ \left(\frac{1}{4}+2\epsilon\right)\lVert\add_j\rVert\Delta + \norm{\qnbhd(\add_j) \cap (\unsat \setminus \qnbhd(\env_{j-1}))}\right]\\
    \leq& \left(\frac{1}{4}+2\epsilon\right)\lVert\env_i\rVert\Delta + \norm{\unsat}~.
\end{align}
We use the fact that the sets $\add_j$ are disjoint, so $\lVert\env_i\rVert = \sum_{j\leq i}\lVert\add_i\rVert$.

This completes the proof.
\end{proof}

\subsubsection{Combining the bounds}
Using the upper bound and lower bound on $\norm{\qnbhd(\env)}$ from Lemma~\ref{lem:max-expanding-size} and Lemma~\ref{lem:upp-L-bound}, we arrive at a bound on $\lVert \env \rVert$.

\begin{lemma}
    \label{lem:combining-twonorm}
    Let $\env$ be the output of Algorithm~\ref{alg:ssfind} on a reduced error $\err$ such that $\Delta(\lVert \err \rVert+1)\frac{4}{1-10\epsilon} \leq \min(\alpha_V \Delta_V n,\alpha_C \Delta_C m)$.
    Then 
    \begin{align*}
        \lVert \env \rVert \leq \lVert \err\rVert \frac{4}{(1-10\epsilon)}~.
    \end{align*}
\end{lemma}
\begin{proof} 
    Assume towards a contradiction that the $\ssfind$ algorithm outputs an envelope $\env$ that is too large, i.e.\
    \begin{align}
        \lVert \env\rVert > \lVert \err \rVert \frac{4}{(1-10\epsilon)}~.
    \end{align}
    Let $\env_j$ be the envelope after the $j$\textsuperscript{th} iteration of the algorithm.
    Let $\env_i$, the envelope in the $i$\textsuperscript{th} iteration, be the largest envelope such that it still obeys $\Delta \lVert \env_i\rVert\leq \min(\alpha_V \Delta_V n,\alpha_C \Delta_C m)$.

    This implies that
    \begin{align}
        |\env_i \cap \qubits_V| \leq \alpha_V n~, |\env_i \cap \qubits_C| \leq \alpha_C m~.
    \end{align}
    From Lemma~\ref{lem:max-expanding-size}, this set expands:
    \begin{align}
        \norm{\qnbhd(\env_i)}\geq \frac{1}{2}(1-\epsilon)\Delta \lVert \env_i \rVert~.
    \end{align}
    From Lemma~\ref{lem:upp-L-bound}, its size is bounded from above:
    \begin{align}
         \norm{\qnbhd(\env_i)}\leq \norm{\unsat} + \left(\frac{1}{4} + 2\epsilon \right)\Delta \lVert \env_i \rVert~.      
    \end{align}
    Together this implies that
    \begin{align}
    \label{eq:twonorm-bound}
        \Delta\lVert \env_i \rVert \leq \norm{\unsat}\frac{4}{(1-10\epsilon)}~.
    \end{align}
    Each unsatisfied parity check has to be adjacent to some error, $\norm{\unsat}\leq \Delta \lVert \err \rVert$.
    The above bound in Eq.\ \eqref{eq:twonorm-bound} translates to 
    \begin{align}\label{eq:l-bound}
        \lVert \env_i \rVert \leq \lVert \err \rVert\frac{4}{(1-10\epsilon)}~.
    \end{align}

    The set $\env_i$ is chosen to be the maximal set such that $\Delta \lVert \env_i\rVert\leq \min(\alpha_V \Delta_Vn,\alpha_C \Delta_C m )$.
    Therefore, $\Delta\lVert \env_{i+1}\rVert > \min(\alpha_V \Delta_V n,\alpha_C \Delta_C m)$.
    The set $\env_{i+1}$ is created by adding to $\env_i$ a reduced set $\add_{i+1}$, which means that $ \Delta\lVert \env_{i+1} \rVert \leq \Delta(\lVert \env_{i} \rVert + 1)$.
    
    Using (\ref{eq:l-bound}) and the bound on $\lVert \err\rVert$ from the lemma statement we get
    \begin{align}
        \Delta\lVert \env_{i+1} \rVert &\leq \Delta(\lVert \env_{i} \rVert + 1)\\
        &\leq  \Delta(\lVert  \err \rVert+1)\frac{4}{(1-10\epsilon)}\\
        &\leq \min(\alpha_V \Delta_V n,\alpha_C \Delta_C m)~,
    \end{align}
    which is a contradiction. Therefore the set $\env$ must satisfy $\lVert\env\rVert\leq 4\lVert \err\rVert/(1-10\epsilon)$.
\end{proof}

\begin{corollary}
\label{cor:combining-onenorm}
    Let $\env$ be the output of Algorithm~\ref{alg:ssfind} on a reduced error $\err$ such that
    \begin{align}
        |\err| \leq \frac{(1-10\epsilon)}{4} \cdot \frac{\Delta_V}{\Delta_C} \cdot \min(\alpha_V n, \alpha_C  m) - \Delta_V~.
    \end{align}
    Then the size of the envelope is bounded:
    \begin{align*}
        \lvert \env \rvert \leq \lvert \err\rvert \cdot \frac{4}{(1-10\epsilon)} \frac{\Delta_C}{\Delta_V}~.
    \end{align*}
\end{corollary}
\begin{proof}
    Consider an error $\err \subseteq \qubits$ such that
    \begin{align}
        |\err| \leq \frac{(1-10\epsilon)}{4} \cdot \frac{\Delta_V}{\Delta_C} \min(\alpha_V n, \alpha_C  m) - \Delta_V~.
    \end{align}
    This implies that
    \begin{align}
        \Delta \lVert \err \rVert &= |\err_V|\Delta_V + |\err_C|\Delta_C\\
        &\leq |\err|\Delta_C\\
        &\leq \frac{(1-10\epsilon)}{4}\min\left(\alpha_V \Delta_V n, \alpha_C \Delta_C m\right)-\Delta~.
    \end{align}
    Equivalently, $\Delta(\lVert \err \rVert +1)(4/(1-10\epsilon)) \leq \min(\alpha_V \Delta_V n, \alpha_C \Delta_C m)$.
    Lemma~\ref{lem:combining-twonorm} guarantees that
    \begin{align}
        \lVert \env \rVert \leq \lVert \err\rVert \frac{4}{(1-10\epsilon)}~.
    \end{align}
    We convert the weighted norm to a standard one using the fact that $\Delta_C \geq \Delta_V$.
    As $\lVert \env \rVert = \frac{\norm{\env_V}}{\Delta_C} + \frac{\norm{\env_C}}{\Delta_V}$, it implies $|\env|\leq \Delta_C \lVert\env \rVert$. 
    Similarly, $\lVert \err \rVert = \frac{\norm{\err_V}}{\Delta_C} + \frac{\norm{\err_C}}{\Delta_V}$, which implies $\Delta_V\lVert \err \rVert \geq \norm{\err}$.
    
    Together, we get
    \[
        |\env|\leq \Delta_C \lVert \env \rVert \leq \frac{\Delta_C}{\Delta_V}|\err|\frac{4}{(1-10\epsilon)}~.
    \]
    
\end{proof}

\subsection{Proof of Theorem~\ref{thm:main}}
The main result, Theorem~\ref{thm:main}, makes two claims.
Informally, these can be stated as follows:
\begin{enumerate}
    \item \textbf{Correctness:} For sufficiently small error $\err$, $\ssfind$ produces an envelope $\env$ that contains the error $\err$ and simultaneously, $\env$ is not much larger than $\err$.
    \item \textbf{Time Complexity:} $\ssfind$ terminates in time $O(|\err|)$.
\end{enumerate}
We formally establish each of these claims below.

\subsubsection[Bound on E]{Correctness}

For the error $\err$ to be entirely covered by the envelope $\env$, Lemma~\ref{lem:coverage} requires
\begin{align*}
    |\err_V| \leq \alpha_V n~, \quad |\err_C| \leq \alpha_C m~.
\end{align*}
On the other hand, to guarantee that $\env$ is bounded, Corollary~\ref{cor:combining-onenorm} requires
\begin{align*}
    |\err| \leq \frac{(1-10\epsilon)}{4} \cdot \frac{\Delta_V}{\Delta_C} \cdot \min(\alpha_V n, \alpha_C  m) - \Delta_V~.
\end{align*}
The latter is the stronger of the two.

\subsubsection{Time complexity}\label{sec:time}
The algorithm is divided into two phases.
In the $\mathtt{Setup}$ $\mathtt{Phase}$, we compute the score of each reduced set $\add \in \minsets(\gens)$.
This involves computing the score of each generator $\gen \in \gens$.
The output is stored in a data structure lookup table $\mathtt{lookup}\_\mathtt{table}$, a dictionary with two keys `$\leq 2\epsilon$' and `$> 2\epsilon$'.
Corresponding to each key, $\mathtt{lookup}\_\mathtt{table[key]}$ is an unsorted array.
The $\mathtt{Setup}$ $\mathtt{Phase}$ therefore takes time proportional to the number of generators which is $\Theta(N)$.
For each generator, we compute the score for every reduced set $\add \in \supp(\gen)$ and there are at most $2^{\Delta_V+\Delta_C}$ such sets.
Therefore, the time complexity scales exponentially in the degree $\Delta_V$ and $\Delta_C$ of the input graph $G$.
However, these are independent of $N$ for the family of LDPC codes.
In fact, this can be made considerably better.
We only need to compute the score of generators that are adjacent to $\unsat$.
The total number of generators that are adjacent to $\unsat$ is at most $O(|\unsat|) = O(\sqrt{N})$.

In the $\mathtt{Main}$ $\mathtt{Phase}$, we add the set $\add$ to $\env$ if $\score(\add) \leq 2\epsilon$.
This takes constant time as $\add$ is itself a constant-sized set.
After adding $\add$ to $\env$ we only need to update the score of reduced sets in generators that are adjacent to $\add$.
There are only a constant number of such generators.
To be precise, there are at most $\Delta\lVert\add\rVert$ generators that are incident to $\add$.
We then need to update whether these generators belong to $\mathtt{lookup}\_\mathtt{table[\leq 2\epsilon]}$ or $\mathtt{lookup}\_\mathtt{table[> 2\epsilon]}$.
As these are unsorted arrays, insertion takes constant time.
Hence the time complexity of each iteration is constant.

In each iteration, the size of the envelope increases by at least one.
Therefore, the total number of iterations is at most $|\env|$.
From Corollary~\ref{cor:combining-onenorm}, this is a function of the size of the error which in turn obeys $|\err| = O(\sqrt{N})$.
The time required to complete the $\mathtt{Main}$ $\mathtt{Phase}$ is the product of the time required for each iteration and the total number of iterations.
It is thus $O(|\err|) = O(\sqrt{N})$.

This completes the proof of Theorem~\ref{thm:main} (and our paper).

\bibliographystyle{alpha}
\bibliography{references}

\appendix
\section{Viderman's algorithm}\label{sec:classic-code}

In this section, we present a proof of Viderman's algorithm for decoding errors on (classical) expander codes~\cite{viderman2013linear}. This is a simpler proof that in the original paper, that achieves slightly worse parameters.

\begin{claim}
    The $\find$ algorithm, when receiving a set $\unsat$ which are the unsatisfied parity checks of a word $\ccorruptcw\in\set{0,1}^n$, such that $\ccorruptcw=\ccodeword+\berr$ for $\ccodeword\in \ccode$ and an error $\berr$ such that $|\set{i|\berr_i=1}|\leq (1-3\epsilon_V)(\alpha_V n-1)$, outputs an envelope $L$ such that $\set{i|e_i=1}\subset L$.
\end{claim}

We begin with an intuitive description of the decoding algorithm. Denote by $\cerr=\set{i|\berr_i=1}$ the set of vertices that has an error.
We refer to $L$ as the envelope and $R = \nbhd(L) \union \unsat$ as the set of untrustworthy parity checks.
In each iteration, $\find$ searches for bits that are adjacent to too many untrustworthy parity checks.
If it finds such a bit $v \in V$, it marks it as suspicious and marks all its neighbors $\nbhd(v)$ as untrustworthy.
When the algorithm terminates, it returns an envelope $L$ which is guaranteed to contain the error $\cerr$.
By erasing all bits in the support of $L$ and running an appropriate erasure decoding algorithm, we can recover the transmitted message.

The proof that $\find$ works correctly proceeds in two phases.
First, they show that $\cerr \subseteq L$, i.e.\ that the envelope $L$ covers the error $\cerr$.
Next, they show that the algorithm will terminate such that $L$ does not encompass the entire set of bits.

\begin{claim}
    Assuming $|\cerr|\leq \alpha_V n$, the $\find$ algorithm returns a set $L$ such that $\cerr\subset L$.
\end{claim}
\begin{proof}
    Let $L$ be the set after $\find$ algorithm finished running, and assume towards a contradiction that $\cerr\not\subseteq L$. Let $B = \cerr\setminus L$, the parts of the error that are not covered by $L$.

    From our assumption $|B|\leq|\cerr|\leq \alpha_V n$, therefore we have that 
    \begin{align}
        |\nbhd(B)|\geq |B|(1-\epsilon_V)\Delta_V.
    \end{align}
    This implies a lower bound on the unique expansion of $B$:
     \begin{align}
        |\unbhd(B)|\geq |B|(1-2\epsilon_V)\Delta_V.
    \end{align}
    Let $v\in B$ be a vertex with $|\unbhd(B)\cap\nbhd(v)|\geq (1-2\epsilon_V)\Delta_V$. Then  every parity check $c\in \unbhd(B)\cap\nbhd(v)$ is adjacent to $v$ that is an error, and not to any other vertices in $B$. If it is adjacent to some $v'\in L$ then by definition $c\in R$. If it is not adjacent to $L$, then it is adjacent to exactly one error, and therefore in $c\in\unsat\subset R$.

    Therefore, $|\nbhd(v)\cap R|\geq (1-2\epsilon)\Delta_V$ and $v$ should have been added to $L$.
\end{proof}

We now give a simple proof that the above algorithm stops. In Viderman's paper \cite{viderman2013linear} there is also a more complicated proof that achieves better parameters. We present the simpler proof to demonstrate the main idea of the stopping condition.
\begin{claim}
    \label{claim:classical-stopping}
    Assume that $|\cerr|\leq (1-3\epsilon_V)(\alpha_V n-1)$.  Then the $\find$ algorithm (Algorithm~\ref{alg:find-gentle}) returns a set $L$ such that $|L|\leq \frac{|\cerr|}{(1-3\epsilon_V)}$.
\end{claim}
\begin{proof}
Let $L_0=\emptyset,L_1,\ldots,L_t$ be the envelopes during the run of the algorithm, and let $v_1,v_2\ldots,v_t$ be the vertices added, i.e. $L_i = \set{v_i}\cup L_{i-1}$. Assume towards a contradiction that the claim does not hold, i.e. that $|L_t|>\frac{|\cerr|}{(1-3\epsilon_V)}$, and let $i$ be the latest iteration such that $|L_i|\leq \alpha_V n$.

The proof uses a lower bound and an upper bound on $\norm{\nbhd(L_i)}$. 
The lower bound is given from the graph expansion. Since $G$ is an $(\alpha_V,\epsilon_V)$ left vertex-expander,
\begin{align}
    \norm{\nbhd(L_i)}\geq (1-\epsilon_V)\Delta_V\norm{L_i}~.
\end{align}

We now show an upper bound on $\norm{\nbhd(L_i)}$. From the algorithm, a vertex $v_j$ is added to $L_{j-1}$ only if $|\nbhd(v_j)\cap R|\geq (1-2\epsilon_V)\Delta_V$. Observe that $R=\nbhd(L)\sqcup(\unsat\setminus\nbhd(L))$, so we can write for every $j\in[t]$
\begin{align}
    |\nbhd(v_j)\cap \nbhd(L_{j-1})|=&|\nbhd(v_j)\cap R| - |\nbhd(v_j)\cap(\unsat\setminus\nbhd(L_{j-1}))|\\\geq& (1-2\epsilon_V)\Delta_V-|\nbhd(v_j)\cap(\unsat\setminus\nbhd(L_{j-1}))|~.
\end{align}
Using this inequality we can bound the expansion of $L_i$,
\begin{align}
    |\nbhd(L_{i})| =& \sum_{j=1}^i |\nbhd(L_{j})\setminus \nbhd(L_{j-1})|\\
    =&\sum_{j=1}^i \left(|\nbhd(v_{j})| - |\nbhd(v_{j})\cap \nbhd(L_{j-1})|\right)\\
    \leq& \sum_{j=1}^i \left(\Delta_V - (1-2\epsilon_V)\Delta_V +  |\nbhd(v_j)\cap(\unsat\setminus\nbhd(L_{j-1}))| \right)\\
    \leq & 2\epsilon_V\Delta_V i + |\unsat|~.
\end{align}

We use the two bounds on $|\nbhd(L_i)|$ to get a bound on $|L_i|$. We note that $|L_i|=i$, as we add a single vertex on each round.
We get
\begin{align}
    (1-\epsilon_V)\Delta_V\norm{L_i}\leq2\epsilon_V\Delta_V |L_i| + |\unsat|~.
\end{align}
This implies a bound on $L_t$,
\begin{align}
    |L_i|\leq \frac{|\unsat|}{(1-3\epsilon_V)\Delta_V} \leq \frac{|\cerr|}{1-3\epsilon_V}~,
\end{align}
using the fact that $|\unsat|\leq\Delta_V|\cerr|$. According to our assumption, $|L_{i+1}| = |L_i| +1 > \alpha_V n$, i.e. $|L_i| >  \alpha_V n-1$, meaning that $|\cerr| > (1-3\epsilon_V)(\alpha_V n-1)$, which is a contradiction.
\end{proof}

As was already highlighted in Viderman's original paper \cite{viderman2013linear}, this algorithm asks less of $G$ when compared to $\flip$---it is sufficient that $\epsilon_V < 1/3$ rather than $\epsilon_V < 1/4$ \cite{sipser1996expander}.

\end{document}